\documentclass[a4paper,10pt]{scrartcl}
\usepackage[utf8]{inputenc}
\usepackage[english]{babel}
\usepackage[fixlanguage]{babelbib}
\usepackage{cite}
\usepackage{amssymb, amsmath, amstext, amsopn, amsthm, amscd, amsxtra, amsfonts, commath}
\usepackage{yfonts, mathrsfs}
\usepackage{colonequals}
\usepackage{enumerate}
\usepackage{graphicx}
\usepackage[all]{xy}
\usepackage{todonotes}
\usepackage{tikz-cd}
\usepackage{hyperref}%
\hypersetup{
urlcolor = red,
colorlinks = true,
linkcolor = blue,
citecolor = blue,
linktocpage = true,
pdftitle = {Lie Theory for Asymptotic Symmetries in General Relativity: The NU Group},
pdfauthor = {David Prinz, Alexander Schmeding},
bookmarksopen = true,
bookmarksopenlevel = 1,
unicode = true,
hypertexnames =false
}%
\usepackage{cleveref}
\swapnumbers
\newtheoremstyle{standard}
{16pt} 
{16pt} 
{} 
{} 
{\bfseries}
{} 
{ } 
{{\thmname{#1~}}{\thmnumber{#2.}}\thmnote{~(#3)}} 
\newtheoremstyle{kursiv}
{16pt} 
{16pt} 
{\itshape} 
{} 
{\bfseries}
{} 
{ } 
{{\thmname{#1~}}{\thmnumber{#2.}}\thmnote{~(#3)}} 
\theoremstyle{standard}
\newtheorem{defn}{Definition}[section]

\newtheorem{rem}[defn]{Remark}

\newtheorem{setup}[defn]{}

\theoremstyle{kursiv}
\newtheorem{thm}[defn]{Theorem}
\newtheorem{prop}[defn]{Proposition}
\newtheorem{cor}[defn]{Corollary}
\newtheorem{lem}[defn]{Lemma}
\newtheorem{conj}[defn]{Conjecture}

\newcommand{\Lf}{\ensuremath{\mathbf{L}}}

\newcommand{\bC}{\ensuremath{\mathbb{C}}}
\newcommand{\eC}{\ensuremath{\hat{\mathbb{C}}}}

\newcommand{\R}{\ensuremath{\mathbb{R}}}

\newcommand{\N}{\ensuremath{\mathbb{N}}}
\newcommand{\Z}{\ensuremath{\mathbb{Z}}}

\newcommand{\SSS}{\ensuremath{\mathbb{S}}}

\newcommand{\coloneq}{\colonequals}
\newcommand{\ev}{\ensuremath{\operatorname{ev}}}
\DeclareMathOperator{\one}{\mathbf{1}}

\DeclareMathOperator{\Diff}{Diff}
\DeclareMathOperator{\Evol}{Evol}

\DeclareMathOperator{\id}{id}

\DeclareMathOperator{\Lor}{\mathrm{SO}^+ (3,1)}

\DeclareMathOperator{\Ncvs}{\mathcal{N}_0^{\mathrm{vs}}}
\DeclareMathOperator{\NUovs}{NU_0^{\mathrm{vs}}}

\DeclareMathOperator{\Aut}{Aut}
\newcommand{\Frechet}{Fr\'{e}chet}

\newcommand{\LB}[1][\cdot \hspace{1pt} , \cdot]{\left[\hspace{1pt} #1 \hspace{1pt} \right]}
\DeclareMathOperator{\evol}{\mathrm{evol}}
\DeclareMathOperator{\BMS}{BMS}
\DeclareMathOperator{\NU}{NU}

\newcommand{\scri}{\mathscr{I}}

\newcommand{\conext}{\widehat{M}}
\newcommand{\conmet}{\hat{g}}
\newcommand{\z}{\mathbf{z}}

\begin{document}
\title{Lie Theory for Asymptotic Symmetries in General Relativity: The NU Group}
\author{David Prinz\footnote{Humboldt-Universität zu Berlin, Germany: \href{mailto:prinz@math.hu-berlin.de}{prinz@math.hu-berlin.de}}\hspace{2mm} and Alexander Schmeding\footnote{Nord universitetet i Levanger, Norway: \href{mailto:alexander.schmeding@nord.no}{alexander.schmeding@nord.no}}}
\date{June 29, 2022}
{\let\newpage\relax\maketitle}

\begin{abstract}
	We study the Newman--Unti (NU) group from the viewpoint of infinite-dimensional geometry. The NU group is a topological group in a natural coarse topology, but it does not become a manifold and hence a Lie group in this topology. To obtain a manifold structure we consider a finer Whitney-type topology. This turns the unit component of the NU group into an infinite-dimensional Lie group. We then study the Lie theoretic properties of this group. Surprisingly, the group operations of the full NU group become discontinuous, whence the NU group does not support a Lie group structure. The NU group contains the Bondi--Metzner--Sachs (BMS) group as a subgroup, whose Lie group structure was constructed in a previous article. It is well known that the NU Lie algebra splits into a direct sum of Lie ideals of the Lie algebras of the BMS group and conformal rescalings of scri. However, the lack of a Lie group structure on the NU group implies that the BMS group cannot be embedded as a Lie subgroup therein.
\end{abstract}



\medskip


\tableofcontents

\section{Introduction and statement of results} 

In \cite{Prinz_Schmeding_1}, we studied the Lie group structure of the Bondi--Metzner--Sachs (BMS) group \cite{Bondi_VdBurg_Metzner,Sachs} from the viewpoint of infinite-dimensional geometry. The BMS group is defined as the subgroup of spacetime-diffeomorphisms that preserve asymptotic flatness in the sense of Bondi, Van der Burg, Metzner and Sachs. Asymptotically flat spacetimes are ideally suited to study gravitational waves, which were first experimentally verified at LIGO in 2016 \cite{Abbott_et-al}, and their corresponding symmetry groups provide useful insight into the gravitational \(S\)-matrix via soft-scattering theorems for the corresponding Feynman rules \cite{Weinberg,Strominger,He_et-al}, cf.\ \cite{Prinz} and the references therein. In this article we extend our investigations to the Newman--Unti  (NU) group \cite{Newman_Unti}, which turns out to be much more involved. On the level of Lie algebras it is a well-known fact that the corresponding Lie algebras of the BMS and NU groups are related via the following direct sum decomposition \cite{Barnich:2009se,Barnich:2010eb}
\begin{equation}
	\mathfrak{nu}_4 \cong \mathfrak{bms}_4 \oplus \mathfrak{conf}_\scri \, , \label{eqn:decomposition_nu-bms-conf}
\end{equation}
where \(\mathfrak{nu}_4\) and \(\mathfrak{bms}_4\) denote the respective Lie algebras for the NU and BMS groups in 4 dimensions of spacetime and \(\mathfrak{conf}_\scri\) denotes the abelian Lie algebra of conformal rescalings on scri \(\scri\), the corresponding null boundary of the spacetime. Interestingly, this decomposition is not compatible with the smooth structure and does therefore not carry over to the global Lie group level.  More precisely, there is a topological obstruction which is similar to the situation for diffeomorphism groups of non-compact manifolds. There, the Lie algebra of the diffeomorphism group needs to be defined as the Lie algebra of compactly supported vector fields and is not equal to the Lie algebra of all vector fields. Unfortunately, this restriction is not compatible with the decomposition of equation~\eqref{eqn:decomposition_nu-bms-conf}. In the following, we will briefly compare the BMS and NU groups from a physical and from a mathematical perspective. Furthermore, we refer to our previous article \cite{Prinz_Schmeding_1} for a general introduction to asymptotic symmetry groups in General Relativity and infinite-dimensional Lie groups.\smallskip

\textbf{Acknowledgments:} The authors wish to thank G.~Barnich and H.~Gl\"{o}ckner for helpful discussions on the subject of this work. Furthermore, we thank H.~Gl\"{o}ckner for access to his manuscript \cite{hgsemidirect}. Additionally, we wish to thank the referees for their careful proofreading and useful feedback. DP is additionally supported by the University of Potsdam.

\subsection{Physical motivation}

Asymptotic symmetry groups capture the symmetries of physical models in General Relativity with a finite matter content. Such models are given e.g.\ by isolated stars, but also complete universes with a finite matter distribution. Therefore, asymptotically flat spacetimes are ideally suited to study gravitational waves - a topic that has regained a lot of interest after the experimental verification of gravitational waves at LIGO in 2016 \cite{Abbott_et-al}. However, there are different approaches to asymptotically flat spacetimes, each of which leading to different asymptotic symmetry groups, see e.g.\ \cite{Friedrich:2017cjg,Ruzz20} for overviews and \cite{Prinz_Schmeding_1} for a general introduction. Specifically, there are coordinate based approaches, such as the original BMS and NU approach \cite{Bondi_VdBurg_Metzner,Sachs,Newman_Unti}, as well as the geometric approach by Penrose \cite{Penrose_1,Penrose_2,Penrose_3,Penrose_4}. We refer to \cite{SaSaW75} for a characterization of the BMS and NU groups in terms of the geometric approach of Penrose. Physically, the BMS group and the NU group differ in the definition of the corresponding `supertranslations', given for 4-dimensional spacetimes with vanishing cosmological constant via the function spaces \(\mathcal{S} \coloneq C^\infty (\SSS^2)\) for the BMS group and \(\mathcal{N} \subset C^\infty (\mathbb{R} \times \SSS^2)\) for the NU group, where the subset is characterized in \Cref{defn:mathcalN}. Geometrically, the difference lies in the action of the corresponding groups on scri \(\scri \coloneq \mathbb{R} \times \SSS^2\) (cf.\ \cite[Subsection 2.1]{Prinz_Schmeding_1} and Subsection~\ref{ssec:Penrose_conformal_extention} of the present article), the null-boundary of the aforementioned spacetimes, and its induced (degenerate) metric. More precisely, the NU group allows for conformal rescalings of the induced metric, whereas the BMS group is additionally subjected to the invariance of a specific (2,2)-tensor, defining a `strong conformal geometry', \cite{Penrose1974,SaSaW75}.

\subsection{Mathematical motivation}

The NU group is known to be an ``infinite-dimensional group''. This directly motivates the question: Is it an infinite-dimensional Lie group? The answer to this depends on the definition of infinite-dimensional Lie group. The present articles concept of an infinite-dimensional Lie group is a group with a manifold structure modeled on a locally convex space, such that the group operations are smooth. As ordinary calculus breaks down beyond the realm of Banach spaces (and the NU group can not be modeled on a Banach space), smoothness here refers to smoothness in the sense of Bastiani calculus (see \cite{S2021} for more information).
Furthermore we refer to \cite[Appendix A]{Prinz_Schmeding_1} for a concise overview on infinite-dimensional Lie groups. Specifically, we recalled that while the Lie group exponential provides a local diffeomorphism (exponential coordinates) between finite-dimensional Lie groups and their Lie algebras, this correspondence breaks down in general in the infinite-dimensional case. As an example, the Lie group exponential of the diffeomorphism group fails to be locally surjective. In this article, we encounter another complication: For the BMS group, it was sufficient to consider function spaces on compact manifolds. The NU group is modeled on spaces of functions on non-compact manifolds. 
This might seem like a minor issue, but it leads to major complications stemming from the function space topologies. Again this complication is well known for the diffeomorphism group of a non-compact manifold \cite{Sch15}. To illustrate this, recall that for a compact manifold \(K\) the Lie algebra of the Lie group $\Diff (K)$ is the Lie algebra of vector fields on $K$ with the negative of the usual bracket of vector fields. The natural topology on this Lie algebra is the compact open $C^\infty$-topology. This topology allows one to simultaneously control a function and up to finitely many of its derivatives on any given compact set.  
The Lie group exponential of $\Diff (K)$ is the map which sends a vector field to its time $1$-flow and it is not difficult to see that it is smooth with respect to the compact open $C^\infty$-topology.

Switching now to a non-compact manifold \(N\), consider the Lie algebra of vector fields on $N$ with the compact-open $C^\infty$-topology. Since $N$ is non-compact, this topology does not control the behavior of vector fields at infinity. Indeed, it is too coarse in the non-compact case. Extrapolating from the compact case, also the Lie group exponential is problematic: On a non-compact manifold there are vector fields whose flow explodes in arbitrarily short time, so the flow map is simply not available on the algebra of all vector fields. 
All of these problems can be solved by switching to the Lie algebra of compactly supported vector fields on $N$ with a finer topology. This topology, called the fine-very strong topology, yields enough control to work with the restricted algebra. Moreover, flows exist now and give rise to a smooth Lie group exponential turning the group $\Diff (N)$ into a Lie group.

Returning to the NU group, the problems we face stem from the well-known difficulties for diffeomorphism groups which were just outlined. The BMS group is a product of an infinite-dimensional abelian group (indeed a vector space of functions on a compact manifold) with a finite-dimensional Lie group. For the NU group, the abelian part is replaced by a parameterized version of the diffeomorphism group of a non-compact manifold. So one should expect that all problems associated to diffeomorphism groups in the passage from compact to non-compact domains are again present.

\subsection{Original results of the present article}

In the present article, we investigate whether the NU group can be made an infinite-dimensional Lie groups in the sense of Milnor \cite{Milnor}. Specifically, we use the so-called \emph{Bastiani calculus}, an introduction to which can be found e.g.\ in \cite[Appendix B]{Prinz_Schmeding_1} and the book \cite{S2021}. Let us first recall the algebraic structure of the NU group. We denote by $\Lor$ the orthochronous Lorentz group and note that it is a finite-dimensional Lie group which acts on the function space $C^\infty (\R \times \SSS^2)$. This action will be recalled later; for now, we only mention that the action restricts to a certain (non-commutative) group $\mathcal{N} \subseteq C^\infty (\R \times \SSS^2)$ and the NU group is then the semidirect product
\begin{equation}\label{eq:NUsemidirect}
\NU = \mathcal{N} \rtimes \Lor  \, .
\end{equation}
There is a canonical identification $I \colon \BMS \rightarrow \NU$ identifying the BMS group as a subgroup of the NU group. Note that the semidirect product \eqref{eq:NUsemidirect} is topologically much more involved compared to the BMS case. The reason is that the functions in $\mathcal{N}$ are defined on a non-compact manifold, while functions in $\mathcal{S}$ are defined on a compact domain. 
The non-compactness of \(\R \times \SSS^2\) implies that $C^\infty (\R \times \SSS^2)$ admits at least three qualitatively different choices of function space topologies (the compact open $C^\infty$-topology, the Whitney topologies and the fine very strong topology). Only two of these topologies provide additionally either a topological vector space structure (compact open $C^\infty$-topology) or a manifold structure (fine very strong topology) on $C^\infty (\R \times \SSS^2)$.

As our first result, we prove that $\mathcal{N}$ and also the NU group becomes a topological group in the compact open $C^\infty$-topology. The inclusion of the BMS group becomes then a morphism of topological groups. The compact open $C^\infty$-topology is too coarse to construct a manifold structure on $\mathcal{N}$, whence there is no Lie group structure to be gained here.

Passing to the fine very strong topology, we construct a manifold structure on $\mathcal{N}$ which turns the group $\mathcal{N}$ into a Lie group. Note that the model space of $\mathcal{N}$ is neither a Banach nor a \Frechet\ space, whence Bastiani calculus is needed to make sense of this. We then identify the Lie algebra of $\mathcal{N}$ and establish that $\mathcal{N}$ is a regular Lie group. Recall that a Lie group is regular if a certain kind of ordinary differential equations can be solved on the Lie group and the solution depends smoothly on parameters, cf.\ e.g.\ the introduction in \cite{Prinz_Schmeding_1}. Regularity is a prerequisite for advanced tools in (infinite-dimensional) Lie theory.
Surprisingly, the action of the Lorentz group on $\mathcal{N}$ is discontinuous with respect to the fine very strong topology, so the NU group cannot be turned into a Lie group if $\mathcal{N}$ carries this topology. However, endowing the function space $\mathcal{N}$ with the fine very strong topology, the action of $\Lor$ on the identity component of $\mathcal{N}$ becomes smooth. To distinguish the identity component in the fine very strong topology from the (larger) identity component in the topological group $\mathcal{N}$ with the compact open $C^\infty$-topology, we will write $\Ncvs$. We prove that $\Ncvs$ forms a Lie group. Then we consider the NU group as a manifold (where the function space part $\mathcal{N}$ is again endowed with the fine very strong topology). In this topology the connected component of identity $\NUovs$ (not to be confused with the larger connected component of the identity in the topological group $\NU$ with the compact open $C^\infty$-topology) becomes a Lie group. 
As a consequence of the semidirect product structure, $\NUovs$ is a regular Lie group. Moreover, we establish for $\Ncvs$ and $\NUovs$ that they are not real analytic Lie groups (whence the Baker--Campbell--Hausdorff series does not yield a model for the group structure). However, our analysis shows that the strong Trotter and the strong commutator properties hold for both groups. 

\subsection{A recap of the Bondi--Metzner--Sachs group}

We have established in \cite{Prinz_Schmeding_1} that the BMS group is an infinite-dimensional Lie group modeled on a \Frechet~space. This follows directly from its semidirect product form \cite{Sachs,McCar72a}
\begin{equation}
\BMS = \mathcal{S} \rtimes \Lor \, ,
\end{equation}
where $\mathcal{S} =C^\infty (\SSS^2)$ is a function space (viewed as an abelian Lie group) and  \(\Lor\) denotes the Lorentz group (which is in particular a finite-dimensional Lie group) together with the smoothness of the group action on the infinite-dimensional space of supertranslations $\mathcal{S}$ \cite{Prinz_Schmeding_1}. We found that the BMS group is regular in the sense of Milnor \cite{Milnor}. As a consequence, we obtain the validity of the Trotter product formula and the commutator formula on the BMS group. 
We remark that these formulae are important tools in the representation theory of Lie groups. Additionally, we found a more surprising results of our investigation: The BMS group is \textbf{not} an analytic Lie group. In a nutshell, the reason for this is that the group product incorporates function evaluations of smooth (but not necessarily analytic) mappings. In particular, this entails that
\begin{itemize}
 \item the well known Baker--Campbell--Hausdorff series does not provide a local model for the Lie group multiplication, and
 \item there can not be a complexification of the BMS group which continues the real BMS group multiplication as a complex (infinite-dimensional) Lie group. (See also \cite{McCar92} on complexifications of the BMS group.)
\end{itemize}
Moreover, we remark that due to this defect either the Lie group exponential does not provide a local diffeomorphism onto an identity neighborhood or the BCH-series can not converge on any neighborhood of $0$ in the Lie algebra.

Finally, we have also discuss the case of the generalized BMS group (or gBMS), which can be identified with the semidirect product $\mathcal{S} \rtimes \Diff (\SSS^2)$, and behaves quite similarly (when it comes to the Lie theory) to the BMS group. 

\section{Asymptotically flat spacetimes}

We start this article with an introduction to different definitions of asymptotic flatness. This was started in \cite[Section 2]{Prinz_Schmeding_1} with a particular emphasis on the coordinate-wise approach of Bondi et al.\ and Sachs. In this article, we briefly recall the geometrical construction due to Penrose and then discuss the coordinate-wise construction of Newman and Unti. We refer to \cite{Ashtekar,Friedrich:2017cjg} for excellent overview articles.

\subsection{Penrose's conformal extension} \label{ssec:Penrose_conformal_extention}

In \cite[Subsection 2.1]{Prinz_Schmeding_1} we have reviewed \emph{Penrose's conformal extension}, which we now briefly recall: In this geometrical approach the `physical spacetime' \((M,g)\) gets embedded into its so-called `conformal extension' \(\big ( \conext,\conmet \big )\), that is the spacetime \((M,g)\) together with a boundary \(\scri\), called scri, that represents the points `at infinity', i.e.\
\begin{equation}
	\conext \cong M \sqcup \scri \, ,
\end{equation}
such that the embedding \(\iota \colon M \to \conext\) is a conformal diffeomorphism. Scri \(\scri\) consists of three-dimensional components representing lightlike infinity \(\scri_\pm\) as well as three points, representing timelike and spacelike infinity. In particular, the two metrics \(g\) and \(\conmet\) are conformally related via \(\iota\), i.e.\
\begin{equation}
	\iota_* g \equiv \varsigma^2 \conmet \, ,
\end{equation}
where \(\varsigma \in C^\infty \big ( \conext \big )\) is a smooth function on the conformal extension. If this construction is possible, the spacetime \((M,g)\) is called asymptotically simple. Furthermore, if in addition the Ricci tensor vanishes in a neighborhood of \(\scri\), the spacetime is called asymptotically empty. Using this construction, the BMS and NU groups can be seen as diffeomorphisms acting on \(\scri\), cf.\ \cite{SaSaW75}. Moreover, we recall two important classical results that characterize the geometry of asymptotically simple spacetimes:
\begin{itemize}
	\item If \((M,g)\) is additionally asymptotically empty, then it is globally hyperbolic \cite[Proposition 6.9.2]{Hawking_Ellis} (and thus parallelizable in 4 dimensions of spacetime, cf.\ \cite[Proposition 2.3]{Prinz_Schmeding_1}).
	\item If the cosmological constant is vanishing (\(\Lambda = 0\)), then the two components representing lightlike infinity \(\scri_\pm \subset \scri\) are both homeomorphic to \(\mathbb{R} \times \mathbb{S}^2\) \cite[Corollary 2]{Newman}.
\end{itemize}
Finally, we also mention the viewpoint from symplectic geometry to the above concepts via the Hamiltonian formalism of General Relativity \cite{Ashtekar_Streubel}.

\subsection{Newman and Unti's coordinate-wise definition} \label{ssec:NU_coortinate_definition}

The original approach to this subject was due to the pioneering works of Bondi, van der Burg and Metzner \cite{Bondi_VdBurg_Metzner} and Sachs \cite{Sachs} via a coordinate-wise definition, which are described in \cite{Prinz_Schmeding_1}. In this article, we focus on the coordinate approach of Newman and Unti \cite{Newman_Unti} in the form of \cite{Barnich:2011ty}, where also an explicit relation between the BMS and NU gauges is discussed. As turned out later using Penrose's conformal extension, asymptotically flat spacetimes are parallelizable. Thus, the following coordinate functions can be defined globally, modulo possible singularities. Furthermore, we remark that they are constructed for spherically symmetric situations, which motivates the use of spherical coordinates for the spatial submanifold:

\begin{defn}[NU coordinate functions] \label{defn:bms_coordinates}
	Let \((M,g)\) be an asymptotically simple spacetime with globally defined coordinate functions \(x^\alpha \colon M \to \mathbb{R}^4\), denoted via \(x^\alpha \equiv (t,x,y,z)\). Then we introduce the so-called NU coordinate functions \(y^\alpha \colon M \to \mathbb{R} \times \mathcal{I} \times \mathbb{S}^2\), where \(\mathcal{I} \subseteq \mathbb{R}\) is an open subinterval, denoted via \(y^\alpha \equiv (u,\varrho,\vartheta,\varphi)\), as follows: The first coordinate \(u\) is fixed and serves as a label for null surfaces, the second coordinate \(\varrho (u)\) is an affine parameter for null geodesics and the two remaining ones \(z^a \equiv (\vartheta,\varphi)\) are angular coordinates. Given these coordinates, the metric in the Newman and Unti approach can be expressed as follows:
	\begin{subequations}
	\begin{align}
	g_{\mu \nu} \dif x^\mu \otimes \dif x^\nu & \equiv W \dif u \otimes \dif u - \left ( \dif u \otimes \dif \varrho + \dif \varrho \otimes \dif u \right ) \\ & \phantom{\equiv} + \varrho^2 h_{a b} ( \dif z^a - U^a \dif u ) \otimes ( \dif z^b - U^b \dif u ) \, ,
	\intertext{where \(h_{a b}\) is the metric on the (deformed) unit sphere, which we decompose as follows}
	\begin{split}
	h_{a b} \dif z^a \otimes \dif z^b & \equiv h^{(1)}_{a b} \dif z^a \otimes \dif z^b + \frac{1}{\varrho} h^{(2)}_{a b} \dif z^a \otimes \dif z^b + \scriptstyle \mathcal{O} \displaystyle (\varrho^{-1}) \, , \label{eqn:asymptotics_metric_unit-sphere}
	\end{split}
	\end{align}
	\end{subequations}
	where \(h^{(1)}_{a b}\) is conformally flat and \(h^{(2)}_{a b}\) is traceless with respect to \(h^{(1)}_{a b}\), i.e.\ \({h^{(1)}}^{a b} h^{(2)}_{a b} = 0\). Here, we have expressed the metric degrees of freedom via a real function on the spacetime \(W \in C^\infty ( M, \mathbb{R} )\), a vector field on the unit sphere \(U \in \mathfrak{X} ( \mathbb{S}^2 )\) and a metric on the unit sphere \(h \in \operatorname{Met} (\mathbb{S}^2 )\).
\end{defn}

\begin{defn}[NU asymptotic flatness] \label{defn:nu_asymptotic_flatness}
	Given the coordinate functions from \defnref{defn:bms_coordinates} with the angular functions \(z^a \equiv (\vartheta,\varphi)\) transformed into stereographic coordinates \(\zeta \coloneq \cot (\vartheta / 2) \exp (\mathrm{i} \varphi)\) and \(\overline{\zeta} \coloneq \cot (\vartheta / 2) \exp (- \mathrm{i} \varphi)\), the spacetime \((M,g)\) is called asymptotically flat in the sense of NU, if the following relations are satisfied \cite{Barnich:2011ty}:
	\begin{align} \label{eqns:fall-off_properties}
	W = - 2 \varrho \partial_u k + 4 \exp (- 2 k) \partial_\zeta \partial_{\overline{\zeta}} k + \mathcal{O} (\varrho^{-1}) \, , \quad U^a = \mathcal{O} (\varrho^{-2})
	 \end{align}
	 and Equation~\eqref{eqn:asymptotics_metric_unit-sphere}, where \(k(u, \vartheta, \varphi)\) is the conformal factor such that \(h^{(1)}_{a b} \dif z^a \otimes \dif z^b \equiv \exp (2 k) \dif \zeta \otimes \dif \overline{\zeta}\).
\end{defn}

\begin{rem}
	We emphasize that an explicit relation between the BMS gauge and the NU gauge has been worked out in \cite[Section 4]{Barnich:2011ty}.
\end{rem}

\section{Lie theory for the NU group}

Asymptotic symmetry groups are subgroups of the diffeomorphism group that preserve the chosen boundary condition and gauge fixing. We focus in this article on the NU group. Our aim is to establish (infinite-dimensional) Lie group structures on the NU group . For readers who are not familiar with calculus beyond Banach spaces, we have compiled the basic definitions in \cite[Appendix A]{Prinz_Schmeding_1} (and we suggest to review them before continuing). Furthermore, as the NU group is an extension of the (perhaps more well known) Bondi--Metzner--Sachs (BMS) group, let us recall some facts and notation on the BMS group from \cite{Prinz_Schmeding_1}.

Let  $\mathcal{S} \coloneq C^\infty (\SSS^2) \coloneq C^\infty (\SSS^2,\R)$ be the abelian group of supertranslations and $\Lor$ the orthochronous Lorentz group $\Lor$. Then $\Lor$ acts by conformal transformations on the sphere $\SSS^2$. Furthermore, we can identify elements in $\Lor$ with M\"{o}bius transformations. Recall that a M\"{o}bius transformation admits a matrix representation $\Lambda_f = \begin{bmatrix} a & b \\ c&d\end{bmatrix}$ for $f \in \Lor$ which acts conformally on the Riemann sphere $\eC \cong \SSS^2$. The conformal factor of the transformation $f$ can thus be expressed as
\begin{equation}\label{conffact}
 K \colon \Lor \times \eC \rightarrow ]0,\infty[ ,\quad K_f(\zeta) \coloneq \frac{1+\lVert \zeta\rVert^2}{\lVert a\zeta + b\rVert^2+\lVert c\zeta + d\rVert^2}, \Lambda_f = \begin{bmatrix} a & b \\ c & d\end{bmatrix}.
\end{equation}
Identifying the sphere $\SSS^2$ with the Riemann sphere, we obtain a smooth group action
$$\sigma \colon \mathcal{S} \times \Lor \rightarrow \mathcal{S} , (f,\alpha) \mapsto K_f(\cdot)^{-1} \cdot \alpha \circ f$$
and the BMS group is the semidirect product $\BMS = \mathcal{S} \rtimes \Lor$ with respect to this action. To spell it out explicitly, the group product of the BMS group is 
  $$(F, \phi) (G,\psi) = (F + \sigma(G,\phi), \phi \circ \psi) = (F + K_\phi^{-1} \cdot G \circ \phi, \phi \circ \psi).$$
In \cite{Prinz_Schmeding_1} it was shown that this structure turns the BMS group into an infinite-dimensional regular Lie group.

\subsection{General constructions}

Before we begin, let us recall several general constructions which will be used throughout the following sections. 
We will encounter spaces of differentiable mappings as infinite-dimensional manifolds. Let us repeat some important definitions and properties of these manifolds. In the following $M,N$ will always denote smooth paracompact manifolds and  $C^k(M,N)$ the set of $k$-times continuously differentiable mappings from $M$ to $N$ where $k \in \N_0 \cup \{\infty\}$.

\begin{setup}[Compact open $C^k$-topology]\label{prelim:ck-top}
 If nothing else is said, we topologize $C^k (M,N)$ with the compact open $C^k$-topology. This is the topology turning the mapping
 $$C^k (M,N) \rightarrow \prod_{\ell \in \N_0 , \ell \leq k} C (T^\ell M, T^\ell N) ,\quad f \mapsto (T^\ell f)_{0 \leq \ell \leq k}$$
 into a topological embedding. Here the sets on the right hand side are topologized with the compact open topology and $T^\ell$ denotes the $\ell$-fold iterated tangent functor.
 Recall from \cite[Appendix A]{AaGaS20} that this topology turns $C^k (M,N)$ into a Banach (resp.~\Frechet) manifold for $k \in \N_0$ (resp.~$k=\infty$) if $M$ is compact and $N$ is a finite dimensional manifold. 
 Moreover, one can prove that if $N$ is a locally convex topological vector space, then also $C^k(M,N)$ is a locally convex  topological vector space with the pointwise operations.
\end{setup}

If $M$ is non-compact, the compact open $C^k$-topology does not control the behavior of mappings at infinity. Moreover, the compact open $C^k$-topology does not turn $C^k (M,N)$ into a manifold if $N$ is not a vector space.
For this reason one introduces the so called fine very strong topology on $C^k (M,N)$. We recall its definition now and refer to \cite{HaS17} for more information.

\begin{setup}[The fine very strong topology]\label{setup: topo:open}
 We now endow $C^\infty (M,N)$ with the so called $\mathcal{FD}$-topology or \emph{fine very strong topology} and write $C^\infty_{\text{fS}} (M,N)$ for the space endowed with this topology.
 This is a Whitney type topology controlling functions and their derivatives on locally finite families of compact sets. Before we describe a basis of the fine very strong topology, we have to construct a basis for the strong topology which we will then refine. To this end, we recall the construction of the so called basic neighborhoods (see \cite{HaS17}). Consider $f$ smooth, $A$ compact, $\varepsilon >0$ together with a pair of charts $(U, \psi)$ and $(V,\varphi)$ such that $A \subseteq V$ and $\psi \circ f \circ \varphi^{-1}$ makes sense.
 Then we use standard multiindex notation to define an \emph{elementary $f$-neighborhood}
 $$\mathcal{N}^{r} \left( f; A , \varphi,\psi,\epsilon \right) \coloneq \left\{\substack{\displaystyle g \in C^\infty (M,N), \quad \psi \circ g|_A \quad\text{ makes sense,}\\ \displaystyle\sup_{\alpha \in \N_0^d, |\alpha|< r }\sup_{x \in \varphi (A)}\lVert \partial^\alpha \psi \circ f \circ \varphi^{-1}(x) - \partial^\alpha\psi \circ g \circ \varphi^{-1}(x)\rVert < \varepsilon}\right\}.$$ 
 A basic neighborhood of $f$ arises now as the intersection of (possibly countably many) elementary neighborhoods $\mathcal{N}^{r} \left( f; A_i , \varphi_i,\psi_i,\epsilon_i \right)$ where the family $(V_i,\varphi_i)_{i\in I}$ is locally finite. We remark that basic neighborhoods form the basis of the very strong topology. To obtain the fine very strong topology, one declares the sets  
 \begin{equation}\label{eq: cpt:nbhd}
 \{g \in C^\infty (M,N) \mid \exists K \subseteq M \text{ compact such that } \forall x \in M\setminus K,\ g(x) =f(x) \} \tag{$\star$}
 \end{equation}
 to be open and constructs a subbase of the fine very strong topology as the collection of sets $\eqref{eq: cpt:nbhd}$ (where $f \in C^\infty (M,N)$) and the basic neighborhoods of the very strong topology.  If $M$ is compact, the fine very strong topology coincides with the compact open $C^\infty$-topology.
 \end{setup}

 \begin{setup}\label{setup:exp:law}
 The fine very strong topology turns $C^\infty (M,N)$ into an infinite-dimensional manifold (cf.\ \cite{Michor80} and \cite{HaS17}).
 If $N = F$ is a locally convex space, the pointwise operations turn $C^\infty_{\text{fS}} (M,\R^n)$ into a vector space. However, for non-compact $M$, $C^\infty_{\text{fS}} (M,F)$ is disconnected, whence it is a manifold but not a locally convex space. The largest locally convex space contained in $C^\infty_{\text{fS}} (M,F)$ is the space of compactly supported maps
 $$C^\infty_c (M,F) \coloneq \{f \in C^\infty (M,F) \mid \exists K \subseteq M \text{ compact, s.t.}\ f|_{M\setminus K} \equiv 0\}$$
 We shall always topologize $C^\infty_c (M,F)$ with the fine very strong topology.
 
 Recall from \cite[Lemma A.10]{AaGaS20} that for $M$ compact and $N$ a finite dimensional manifold the manifold $C^k(M,N)$ is canonical, i.e.~a mapping 
 $$h \colon A \rightarrow C^k (M,N) \text{ for any smooth manifold } A$$
 is of class $C^\ell$ if and only if the adjoint map $h^\wedge \colon A \times M \rightarrow N, (a,m)\mapsto h(a)(m)$ is a $C^{\ell,k}$-map.
 This means that $h^\wedge$ is $\ell$-times continuously differentiable with respect to the $A$-component of the product and each of these differentials is then $k$-times differentiable with respect to the $M$-component. This is an extremely useful property of the compact open $C^k$-topology. We warn the reader that the corresponding statement is false for the fine very strong topology if $M$ is non-compact.
 \end{setup}

Finally, we recall from \cite[Lemma 2.2.3]{HaN12} the concept of a semidirect product (all asymptotic symmetry groups in this article will turn out to be semidirect products).

\begin{setup}[Semidirect product of groups]
 Let $N$ and $H$ be groups and $\Aut(N)$ the group of automorphisms of $N$. Assume that $\delta \colon H \rightarrow \Aut (N)$ is a group homomorphism. Then we define a multiplication on $N \times H$ by
 \begin{align}\label{semidirect}
  (n,h) (m,g) := (n\delta(h)(m), hg).
 \end{align}
 This multiplication turns $N \times H$ into a group denoted by $N \rtimes_\delta H$, where $N \cong N \times \{e\}$ is a normal subgroup and $H \cong \{e\}\times H$ is a subgroup. Furthermore, each element $x \in N \rtimes_\delta H$ has a unique representation as $x = nh, n \in N, h \in H$.
 
 If $H,N$ are Lie groups (or analytic Lie groups) and $\delta^\wedge \colon H \times N \rightarrow H, (h,n) \mapsto \delta (h)(n)$ is smooth (analytic)\footnote{If $N$ is finite dimensional, it suffices to require that $\delta$ is a Lie group morphism. In general, there is no Lie group structure on $\Aut (N)$ which guarantees smoothness of the group operation in $N \rtimes_\delta H$.}, then $N \rtimes_\delta H$ is a Lie group (analytic Lie group). Its Lie algebra is a semi-direct product of Lie algebras.
\end{setup}

\subsubsection*{Almost local mappings}
To establish smoothness of certain mappings with respect to the function space topologies just defined we need Gl\"ockner's concept of almost local mappings, see \cite{Glo05}. We present here a version of this technique which allows for parameter dependent almost local mappings. These results were communicated to us by H.~Gl\"{o}ckner and will appear in \cite{hgsemidirect}. We remark here that the proofs for these results are variants of the proofs for the statements without parameter in \cite{Glo05}.

\begin{setup}\label{setup:almostlocal}
 Let $M,N$ be finite-dimensional smooth manifolds and $E,F$ be locally convex spaces. Fix an open set $\Omega \subseteq C^\infty_c (M,E)$. Furthermore, consider a smooth (possibly infinite-dimensional) manifold $P$ and a map 
 $$f \colon P \times \Omega \rightarrow C^\infty_c (N,F).$$
 The map $f$ is called an \emph{almost local} if for each $p \in P$ there exist an open $p$-neighborhood $Q \subseteq P$, a locally finite cover $(U_n)_{n\in A}$ of $M$ by relatively compact open subsets and a locally finite cover $(V_n)_{n\in A}$ of $N$ by open relatively compact subsets (where both families are indexed by the same set $A$) such that the following condition holds
 $$\forall n\in A, \forall q\in Q, \forall \sigma,\tau \in \Omega,\quad \sigma|_{U_n} = \tau|_{U_n} \Rightarrow f(q,\sigma)|_{V_n} = f(q,\tau)|_{V_n}.$$
 If each $(p,\sigma) \in P \times \Omega$ has an open neighborhood $P_0\times \Omega_0$ such that $f|_{P_0\times \Omega_0}$ is almost local, $f$ is called \emph{locally almost local}.
 \end{setup}
 
This notion is relevant due to Gl\"{o}ckner's smoothness proposition for locally almost local maps (which is a parameter-dependent version of \cite[Theorem 3.2]{Glo05}). 

\begin{prop}\label{Glo:paramsmooth}
 In the situation of \Cref{setup:almostlocal}, assume that $P$ is a finite dimensional manifold and the map 
 $$f \colon P\times \Omega \rightarrow C^\infty_c (N,F)$$
 has the following properties:
 \begin{enumerate}
  \item The restriction of $f$ to a mapping 
  $$P \times (\Omega \cap C^\infty_K (M,E)) \rightarrow C^\infty_c (N,F)$$
  is smooth for each compact subset $K \subseteq M$; and 
  \item $f$ is an almost local map (or locally almost local).
 \end{enumerate}
Then $f$ is smooth.
\end{prop}

\begin{rem}
 The statement of \Cref{Glo:paramsmooth} simplifies several assumptions and the result obtained. We mention that \cite{hgsemidirect} establishes in particular a version for finite orders of differentiability and works not only with the spaces $C^\infty_c (M,E)$ but more generally with $C^k$-sections of locally convex vector bundles.
\end{rem}

\subsection{The Newman--Unti group}
Let us first recall the definition of the Newman--Unti group as a certain semidirect product. Let $\Diff^+ (\R)$ be the group of all smooth orientation-preserving diffeomorphisms of $\R$ (recall that a diffeomorphism of $\R$ is orientation-preserving if it has positive derivative everywhere).

\begin{defn} \label{defn:mathcalN}
 Define $\mathcal{N} \coloneq \left\{F \in C^\infty (\R \times \SSS^2) \mid F(\cdot, \z) \in \Diff^+(\R) , \quad \forall \z \in \SSS^2\right\}$. Then $\mathcal{N}$ becomes a group with respect to the product
 $$F \cdot G (u,\z) \coloneq F(G(u,\z),\z).$$
 The unit of the product is the map $p \colon \R \times \SSS^2 \rightarrow \R, (t,\z) \mapsto t$ and the inverse $F^{-1}\colon \R \times \SSS^2\rightarrow \R$ is for $\z \in \SSS^2$ given by $F(\cdot, \z)^{-1}$, where the inverse is computed in $\Diff^+(\R)$. Note that the inverse is the unique smooth map\footnote{Note that smoothness of the inverse in all variables is guaranteed by the implicit function theorem.} satisfying the implicit equation 
 \begin{align}\label{NU:imp:eq}
  t = F(F^{-1}(t,\z),\z) \qquad (t,\z) \in \R \times \SSS^2.
 \end{align}
 Furthermore, using the conformal factor from \eqref{conffact}, the mapping 
 $$\tau \colon  \mathcal{N} \times  \Lor \rightarrow \mathcal{N} , (F, \phi ) \mapsto \left( (t,\z) \mapsto K_\phi^{-1}(\z) F(K_\phi(\z) t,\phi(\z))\right),$$
 makes sense. A quick calculation shows that it is a group action which induces a group morphism $\hat{\tau} \colon \Lor \rightarrow \text{Aut} (\mathcal{S}), \phi \mapsto \tau (\cdot, \phi)$ and we can form the semidirect product of groups. The \emph{Newman--Unti (NU) group} is now the semidirect product
 $$\NU \coloneq \mathcal{N} \rtimes_\tau \Lor.$$
\end{defn}
Explicitly, the multiplication $(F,\phi) \cdot (G,\psi)$ of the NU group is given by 
$$(F(\tau(G,\phi)(t,\z),\z),\phi \circ \psi (\z)) = (F(K_\phi^{-1} (\z)G(K_\phi(\z)t,\phi(\z)),\z),\phi(\psi(\z))).$$
By construction of the group structures we obtain an injective group morphism:
\begin{align}\label{grp:injection}
 I \colon \BMS \rightarrow \NU,\quad (F , \phi) \mapsto \left( (t,\z) \mapsto t + F(\z), \phi\right)
\end{align}
\begin{rem}
 Our definition of the NU group might look odd to the reader used to the usual presentations in the physics literature. 
 In \cite[Definition 5.1]{Alessio:2017lps} the definition of the space $\mathcal{N}$ seems only to require that $F \in C^\infty (\R \times \SSS^2)$ satisfies 
 \begin{align}\label{eq:tooweak}
  \frac{\partial F}{\partial t} (t,\z) > 0 \qquad \forall (t,\z) \in \R \times \SSS^2
 \end{align}
 (an impression one could also have from the formulation in \cite{SaSaW75}, though a second glance shows that the maps need to induce diffeomorphisms of $\scri$); but we note that this definition would \textbf{not} lead to a group structure on $\mathcal{N}$. As an example, we could consider the mapping $F \colon \R \times \SSS^2 \rightarrow \R , F(t,\z) \coloneq \arctan (t)$ which would satisfy \eqref{eq:tooweak} but its inverse $\tan$ can not be extended to a smooth map on $\R$. 
\end{rem}

Compared to the BMS group, one replaces the abelian group of supertranslations by the non-abelian group $\mathcal{N}$. Beyond the non-abelian structure there is another significant difference between the two groups. Whereas the supertranslations are smooth mappings on a compact domain, the group $\mathcal{N} \subseteq C^\infty (\R \times \SSS^2)$ consists of smooth mappings on a non-compact domain. Thus, for $\mathcal{S}$, there is in principle only one choice for the function space topology, while for $\mathcal{N}$ the different topologies do not coincide.\footnote{The space $C^\infty (\R \times \SSS^2)$ supports a variety of inequivalent function space topologies beyond the compact open $C^\infty$ and the fine very strong topology. For example one can define several Whitney-type topologies, see \cite[Section 4]{Michor80}. As these additional topologies do not turn $\mathcal{N}$ into a manifold we shall ignore them.} Before we investigate this behavior, let us note for later use that elements in $\mathcal{N}$ are proper maps:

\begin{lem}\label{lem:properness}
 If $F \in \mathcal{N}$, then for every $K \subseteq \R$ compact, $F^{-1}(K)$ is compact, i.e.~$F$ is a proper map.
\end{lem}

\begin{proof}
 Let $K \subseteq \R$ be compact and let us show that the preimage $F^{-1}(K)$ is compact. To this end, pick a sequence $(r_n,\z_n) \in F^{-1}(K), n \in \N$. We need to show that this sequence has a convergent subsequence. Due to the compactness of $K$, we see that (by passing to a subsequence) we may assume that $F(r_n,\z_n)$ converges in $K$ towards a limit, say $z$. Now $\SSS^2$ is compact, whence we can pass again to a subsequence and assume that $\z_n$ converges towards some $\z_\ast \in \SSS^2$. Now $F(\cdot,\z)$ is a diffeomorphism, whence $r_\ast \coloneq F(\cdot,\z_\ast)^{-1}(z)$ exists. We claim that now $r_n \rightarrow r_\ast$. To see this, consider the diffeomorphisms
 $$G_n \coloneq F(\cdot,\z_\ast)^{-1} \circ F(\cdot , \z_n), \quad n \in \N.$$
 By construction this sequence of diffeomorphisms converges (pointwise) towards the identity and satisfies $G_n (r_n) \rightarrow r_\ast$ as $n \rightarrow \infty$. We see that $G_n^{-1}(r_\ast) \rightarrow r_\ast$ and $r_n - G_n(r_\ast) \rightarrow 0$, whence also $r_n$ converges towards $r_\ast$. We conclude that $(r_n,\z_n)$ has a convergent subsequence and thus $F^{-1}(K)$ is compact. Since $K$ was arbitrary, $F$ is a proper map.  
\end{proof}

We now topologize the Newman--Unti group with respect to the compact open $C^\infty$ topology and consider whether it becomes a topological or Lie group in that way. This is the coarsest function space topology and the NU group with this topology can be analyzed similarly to the BMS group.

\begin{prop}\label{prop:NUtopgp}
 Endow $\mathcal{N}$ with the subspace topology induced by the compact open $C^\infty$-topology on $C^\infty (\R\times \SSS^2)$.
 Then the NU group becomes a topological group.
\end{prop}

\begin{proof}
 From \Cref{setup:exp:law} we see that a map $f \colon M \rightarrow C^\infty (\R \times \SSS^2)$ is smooth if and only if the associated map 
 $f^\wedge \colon M \times (\R \times \SSS^2) \rightarrow \R, (m,(t,\z)) \mapsto f(m)(t,\z) \text{ is smooth.}$
 This entails that the evaluation map $\ev \colon C^\infty (\R \times \SSS^2) \times (\R \times \SSS^2) \rightarrow \R, (f,t,\z)\mapsto f(t,\z)$ is smooth, as it is the associated map to the identity $C^\infty (\R \times \SSS^2) \rightarrow C^\infty (\R \times \SSS^2)$. In the following, we prove directly that group product and inversion in the NU group are continuous (it will be clear from the proof that this entails that also $\mathcal{N}$ is a topological group with the subspace topology and the group action of $\Lor$ on $\mathcal{N}$ is continuous).
 
 \textbf{Step 1:} \emph{The group product of the NU group is continuous}.
 Let us first extend the group product to the locally convex space $C^\infty (\R \times \SSS^2)$ containing $\mathcal{N}$. To write it in terms of composition of mappings, consider $\text{pr}_2 \colon \R \times \SSS^2 \rightarrow \SSS^2, (t,\z) \mapsto \z$. Then we can write the product as: 
 \begin{align}
  P \colon ( C^\infty (\R\times \SSS^2) \times \Lor (\bC))^2 &\rightarrow  C^\infty (\R\times \SSS^2) \times \Lor, \notag\\ 
  ((F,\phi),(G,\psi)) &\mapsto (F \circ (\hat{\tau} (\phi) (G),\text{pr}_2), \phi \circ \psi). \label{extmult}
 \end{align}
 As $\Lor$ is a Lie group, it suffices to establish continuity of the first component of $P$. Note that $\psi$ is irrelevant to the first component of $P$, whence due to the exponential law it suffices to establish smoothness of the mapping 
 \begin{align*}
   C^\infty (\R\times \SSS^2) \times \Lor \times C^\infty (\R \times \SSS^2) \times \R \times \SSS^2 &\rightarrow \R \\
   (F,\phi,G,t,\z) &\mapsto F(K_\phi^{-1} (\z) G(K_\phi (\z) \cdot t , \phi (\z)),\z).
 \end{align*}
However, this mapping is easily seen to be smooth as it can be written as a composition of the smooth mappings $K$, $\delta$ and the smooth action of $\Lor$ on $\SSS^2$ and the smooth evaluation map $\ev$. We deduce from the exponential law that the map $P$ on $C^\infty (\R \times \SSS^2)$ is smooth. In particular, it is continuous and restricts to the group product of the NU group on $\mathcal{N}$. As $\mathcal{N}$ carries the subspace topology, this establishes continuity of the group product.
 
 \textbf{Step 2:} \emph{Inversion in the NU group is continuous}. 
 Recall that in a semidirect product the inversion is given by the formula
 $$\iota \colon \mathcal{N} \times \Lor \rightarrow \mathcal{N} \times \Lor , (F, \Lambda) \mapsto (\tau (F^{-1},\Lambda^{-1}),\Lambda^{-1}))$$
 where $F^{-1}$ is the inverse in the group $\mathcal{N}$. Since $\Lor$ is a Lie group, inversion is smooth and we may suppress it in the following computations. Having endowed $\mathcal{N}$ with the subspace topology of $C^\infty (\R \times \SSS^2)$, continuity of $\iota$ will follow if we can establish the continuity of the mapping 
 $$i_1 \colon \mathcal{N} \times \Lor \rightarrow C^\infty (\R \times \SSS^2), (F,\phi) \mapsto \tau (F^{-1},\phi).$$
 We apply now the exponential law \cite[Theorem B]{AS15} in its strong form: The map $i_1$ is continuous if and only if the associated map 
 $$i_1^\wedge \colon (\mathcal{N} \times \Lor) \times (\R \times \SSS^2) \rightarrow \R, ((F,\phi), (t,\z)) \mapsto K_\phi^{-1}(\z)F^{-1}(t, \phi (\z))$$
 is a $C^{0,\infty}$-mapping. Recall that a mapping is of class $C^{0,\infty}$ if it is continuous and infinitely often continuously differentiable with respect to the second component (i.e.~in the present case, with respect to $(t,\z)$). Rewriting the latter formula we see that 
 $$i_1^\wedge(F,\phi)(t,\z) = \ev (\delta (\phi), \z) \cdot \ev (F^{-1}, (t, \phi (\z)).$$
 Now $\delta$ and the evaluation maps are smooth and we know that $\Lor$ acts smoothly on $\SSS^2$. From the chain rule for $C^{r,s}$-mappings \cite[Lemma 3.17 and 3.18]{AS15} we deduce that $i_1^\wedge$ is a $C^{0,\infty}$-map if the map $H \colon \mathcal{N} \times (\R \times \SSS^2) \rightarrow \R, (F , (t,\z)) \mapsto F^{-1} (t,\z)$ is a $C^{0,\infty}$-map.
 By construction, $H (F, \cdot)$ solves the implicit equation \eqref{NU:imp:eq} for the function $F \in \mathcal{N}$. Hence we can treat the whole equation as an implicit equation with parameter $F$. 
 It is well known that the implicit function theorem with (continuous) parameter yields a smooth solution which depends continuously on the parameter. In other words the map $H$ will be of class $C^{0,\infty}$ as a result of the implicit function theorem. Thus we can deduce continuity of the inversion in the NU group from an implicit function theorem with parameter in a locally convex space (a suitable version of the theorem is recorded in \cite[Proposition 2.1.]{hg06}). 
 \end{proof}
 
 \begin{lem}\label{lem:BMS_into_NU_topological}
  The canonical inclusion $I(F , \phi) = ((t,\z) \mapsto (t + F (\z), \phi)$, \eqref{grp:injection} of the BMS group into the NU group is a topological group morphism.
 \end{lem}
 
 \begin{proof}
 From the definition of $I$ it is clear that it will be continuous if its first component 
 $$\mathcal{S} \rightarrow \mathcal{N},\quad F\mapsto p+F \circ \text{pr}_2$$
 is continuous, where $\text{pr}_2 \colon \R \times \SSS^2 \rightarrow \SSS^2$ denotes the projection onto the second component. We exploit that $\mathcal{N}$ carries the subspace topology of $C^\infty (\R \times \SSS^2)$ and the exponential law \cite[Theorem B]{AS15}. It suffices to prove that the mapping
 \begin{align}\label{adj:formula}
  C^\infty (\SSS^2) \times (\R \times \SSS^2) \rightarrow \R,\quad (F,(t,\z)) \mapsto t + F (\z)
 \end{align}
 is a $C^{0,\infty}$-map. Indeed this map is even smooth as \cite[Proposition 3.20]{AS15} shows that the evaluation map $\ev$ on $C^\infty (\SSS^2)$ is smooth. Hence $(t,\z) \mapsto t + \ev(F,\z)$ and thus also $I$ are smooth.
\end{proof} 
 
\begin{rem}
Note that, for the compact open $C^\infty$-topology, there does not seem to be any sensible way to turn $\mathcal{N}$ into a submanifold. Indeed in the compact open $C^\infty$-topology $\Diff^+ (\R)$ is not a submanifold of $C^\infty (\R)$ in any suitable sense. Thus, the compact open $C^\infty$-topology does not seem to provide a suitable structure to turn $\mathcal{N}$ into an infinite-dimensional Lie group.
\end{rem}
 While the compact open $C^\infty$-topology is too coarse to turn $\mathcal{N}$ into a manifold, switching to a finer topology this problem can be remedied. We will investigate the resulting structure in the next section.
 
 \subsection{The group $\mathcal{N}$ in the fine very strong topology}\label{sect:normalsubgrp}
 In this section we endow $C^\infty (\R \times \SSS^2)$ with the fine very strong topology. We will see that $\mathcal{N}$ is an open subset of $C^\infty_{\mathrm{fS}} (\R \times \SSS^2)$ and this endows $\mathcal{N}$ with a manifold structure as $C^\infty (\R \times \SSS^2)$ is a manifold (cf.\ \cite{Michor80}). The following result seems to be new (and while the related set $\Diff^+ (\R ) \subseteq C^\infty_{\mathrm{fS}} (\R)$ is open (cf.~\cite[Section 10]{Michor80}) this is not immediately useful to establish the following). 

\begin{prop}\label{prop:tech_NUgroup}
 The set $\mathcal{N}$ is an open subset of $C^\infty_{\mathrm{fS}} (\R \times \SSS^2)$.
\end{prop}
The proof of \Cref{prop:tech_NUgroup} is postponed to \Cref{app:aux}. Having now a manifold structure at our disposal, the group $\mathcal{N}$ turns out to be a Lie group.

\begin{prop}\label{prop:NisLie}
 The submanifold structure turns $\mathcal{N} \subseteq C^\infty_{\mathrm{fS}} (\R \times \SSS^2)$ into a Lie group.
\end{prop}

\begin{proof}
 Recall from \cite{HaS17} that in the fine very strong topology $C^\infty_{\text{prop}}(\R \times \SSS^2)$, the subset of proper mappings, is an open subset. By \Cref{lem:properness} and \Cref{prop:tech_NUgroup}, the group $\mathcal{N}$ is an open subset of $C^\infty_{\text{prop}} (\R \times \SSS^2)$ whence it is an open submanifold of $C^\infty_{\text{fS}} (\R\times \SSS^2)$ (see \cite[Theorem 10.4]{Michor80} for a construction of the manifold structure on the function space).
 
 \textbf{Step 1:} \emph{The multiplication is smooth}. We rewrite the group product of $\mathcal{N}$ as follows
 \begin{align}
  F\cdot G = F \circ (G \times \id_{\SSS^2}) \circ (\id_\R \times \Delta) = (\id_\R \times \Delta)^* (\text{Comp} (F,G \times \id_{\SSS^2})) , \label{eq:comp:NU}
 \end{align}

 where $\Delta \colon \SSS^2 \rightarrow \SSS^2\times \SSS^2 , \z \mapsto (\z,\z)$ is the diagonal map, and we have 
 \begin{align*}
  (\id_\R \times \Delta)^* &\colon C^\infty (\R \times \SSS^2, \R \times \SSS^2) \rightarrow C^\infty (\R \times \SSS^2 \times \SSS^2,\R \times \SSS^2), \phi \mapsto \phi \circ (\id_\R \times \Delta)\\
  \text{Comp} &\colon C^\infty (\R \times \SSS^2) \times C^\infty_{\text{prop}} (\R \times \SSS^2 \times \SSS^2, \R \times \SSS^2), (F,\phi) \mapsto F \circ \phi
 \end{align*}
 (where the subscript prop denotes the open subset of proper mappings). From \cite[Corollary 10.14 and Theorem 11.4]{Michor80} we deduce that $(\id_\R \times \Delta)^*$ and $\text{Comp}$ are smooth. Furthermore, the map $C^\infty (\R\times \SSS^2) \rightarrow C^\infty (\R \times \SSS^2 \times \SSS^2,\R \times \SSS^2), G \mapsto  G\times \id_{\SSS^2}$ is smooth by compactness of $\SSS^2$ and \cite[Corollary 11.10 1.]{Michor80}. We deduce from \eqref{eq:comp:NU} that the group product is smooth as a composition of smooth mappings.

 \textbf{Step 2:} \emph{$\mathcal{N}$ is a Lie group}. Since we already know that the group product is smooth and inversion $I \colon \mathcal{N} \rightarrow \mathcal{N}$ is a group anti-morphism, the formula $I(F) = G^{-1} \cdot I(F\cdot G)$ shows that it suffices to prove that $I$ is smooth in an open neighborhood of the identity element $p \colon \R \times \SSS^2 \rightarrow \R, (x,\z) \mapsto x$. In \Cref{localsetup} we construct a local model $\iota$ for the pointwise inversion $I$ on $\Phi (G_0) \subseteq C_c^\infty (\R \times \SSS^2)$ where $\Phi$ is a chart and $G_0$ a suitable $p$-neighborhood. We will now use the auxiliary results from \Cref{app:aux} to establish smoothness of $\iota$.
 
 Then \Cref{lem:res:smooth} shows that $\iota$ restricts to a smooth map on the closed subspaces $C^\infty_L (\R \times K)$. Define now for $R> 0$ the set $O_r \coloneq \{\gamma \in \Phi(G_0) \mid \sup_{(x,k)} |\gamma(x,k)|<R\}$. By construction $O_R$ is an open subset of $\Phi (G_0) \subseteq C^\infty_c (\R \times K)$ and we see that $\Phi(G_0) = \bigcup_{R>0} O_R$. Moreover, \Cref{lem:almostloc} implies that for every $R>0$, the restriction of $\iota$ to $O_R \cap \Phi(G_0)$ is almost local. Hence, Gl\"{o}ckner's smoothness theorem \cite[Theorem 3.2]{Glo05} shows that $\iota$ is smooth, whence $I$ is smooth on $G_0$.
 
 In conclusion, the manifold structure turns the group operations into smooth maps, whence $\mathcal{N}$ is an infinite-dimensional Lie group.
\end{proof}

In essence, the group $\mathcal{N}$ is a version of the group $\Diff^+ (\R)$ with an added parameter which is not directly visible in the composition of the group. Hence the Lie algebra should be given by a Lie algebra of (compactly supported) vector fields where every vector fields depends on a parameter which is not relevant for the Lie bracket. The next proposition shows that this is indeed the case. To formulate it, we recall from \Cref{setup:subspaces} that the set of all compactly supported mappings $C^\infty_c (\R \times \SSS^2)$ becomes a locally convex space if we endow it with the subspace topology induced by the strong very fine topology.

\begin{prop}\label{prop:LieAlg-N}
 The Lie algebra of $\mathcal{N}$ can be identified as $\Lf (\mathcal{N}) = C^\infty_c (\R\times \SSS^2)$ with the Lie bracket given by 
 $$\LB[X,Y] (u,\theta) \coloneq -\LB[X(\cdot,\theta),Y(\cdot,\theta)](u),\quad (u,\theta)\in \R \times \SSS^2$$
 and the Lie bracket on the right is the usual Lie bracket of vector fields (on $\R$).
\end{prop}

\begin{proof}
 The unit in $\mathcal{N}$ is the projection $p (u,\theta)=u$ and we use \cite[10.12]{Michor80} to identify 
 $$T_p \mathcal{N} = \{(p,X) \in C^\infty (\R\times \SSS^2,T\R^2)\mid X \in C^\infty_c (\R\times\SSS^2)\}.$$
 This identifies the Lie algebra as a locally convex space. Now we need to compute the Lie bracket. 
 Again the proof is a variant of the classical argument by which the Lie algebra of the diffeomorphism group can be identified, \cite[6.]{Milnor}. 
 Take $X \in \Lf (\mathcal{N}) = C_c^\infty (\R\times \SSS^2)$. Since multiplication in $\mathcal{N}$ is given by composition in the $u$-component, we deduce from \cite[Corollary 11.6]{Michor80} that the extension of $X$ to a right invariant vector field $R_x$ on $\mathcal{N}$ is given by the formula
 $$R^X (F) = X \circ (F \times \text{pr}_2), \quad F \in \mathcal{N}$$
 where $\text{pr}_2 \colon \R \times \SSS^2 \rightarrow \SSS^2$ is the projection onto the second component. Consider the vector field $R^X \times 0 \in \mathcal{V} (\mathcal{N} \times (\R \times \SSS^2))$. We shall now prove that $R^X \times 0$ is related to $X_0 \colon \R \times \SSS^2 \rightarrow \R \times T\SSS^2, (u,\z) \mapsto (X(u,z),0(\z))$, where $0(\z)$ is the zero-vector field on $\SSS^2$.
 To this end, consider the left action 
 $$\alpha \colon \mathcal{N} \times (\R \times \SSS^2) \rightarrow \R \times \SSS^2,\quad (F,(u,\z)) \mapsto (F(u,\z),\z).$$
 As the evaluation $\ev \colon C^\infty_{fS} (\R \times \SSS^2) \times (\R \times \SSS^2) \rightarrow \R, (F,(u,\z)) \mapsto F(u,\z)$ is smooth by \cite[Corollary 11.6]{Michor80}, it is easy to see that $\alpha$ is a Lie group action. Moreover, \cite[Corollary 11.6]{Michor80} implies that $T_{(F,(u,\z))}\ev (X,Y) = X(u,\z) + dF(u,\z;Y)$. Plugging the vector field $R^X \times 0$ into this formula, we deduce that 
 $$T \alpha \circ (R^X \times 0)(F,(u,\z)) = (X(F(u,\z)) ,0_{\z}) = X_0 (\alpha (F,(u,\z))).$$
 As relatedness is inherited by the Lie bracket of vector fields, we find
 \begin{align*}
 \left(\LB[R^X,R^Y](p)(u,\z), 0(\z)\right) &= T\alpha \left(\LB[R^X\times 0 ,R^Y \times 0]\right)(p,(u,\z)) \\ &= \left(\LB[X(\cdot,\z),Y(\cdot,\z)](u), 0 (\z)\right),  
 \end{align*}
 where the bracket on the right hand side is the usual bracket of vector fields. By definition, the Lie bracket on $\Lf (\mathcal{N})$ is given as $\LB[X,Y] = -\LB[R^X,R^Y](p)$ with the sign shift arising due to the computation with right invariant fields. This proves the claim.
\end{proof}

Having identified the Lie algebra, the next step to develop the Lie theory of $\mathcal{N}$ is to establish the regularity of this Lie group. Let us briefly recall the concept of a regular Lie group.  
 A Lie group $G$ is $C^r$-semiregular, $r\in \N_0 \cup \{\infty\}$, if for every $C^r$-curve $\gamma \colon [0,1] \rightarrow \Lf (G))$ the initial value problem 
  \begin{equation} \label{eq:regularity}\begin{cases}
     \eta'(t) =  T_{\one} \rho_{\eta (t)} (\gamma(t)) \qquad \rho_{g}(h):= h g\\
     \eta (0) = \one
    \end{cases}
 \end{equation}
 has a unique $C^{r+1}$-solution $\Evol (\gamma) := \eta \colon [0,1] \rightarrow G$. If moreover, the evolution map 
 $\evol \colon C^r ([0,1],\Lf (G))  \rightarrow G, \gamma \mapsto \Evol (\gamma)(1)$ is smooth, then $G$ is said to be $C^r$-regular. If $G$ is $C^\infty$-regular (the weakest of the regularity conditions), $G$ is called \emph{regular (in the sense of Milnor)}. To employ advanced techniques in infinite-dimensional Lie theory, one needs to require regularity of the Lie groups involved, cf.\ \cite{hg2015}. Note that for a constant curve $\eta(t) \equiv v\in \Lf(G)$, we simply recover the Lie group exponential $\evol (\eta (t)) = \exp (v)$. Thus every regular infinite-dimensional Lie group admits a Lie group exponential.

 Again, since the regularity of the group $\Diff^+ (\R)$ is a well-known fact (see e.g.\ \cite{KaM97,Sch15} for proofs in the convenient and in the Bastiani setting), it is not hard to imagine that these proofs carry over to $\mathcal{N}$ (as we have in principle just added another parameter to the construction). As the modification is again not trivial, we supply the necessary details now.
 
 Let us first take a look at the differential equation we need to solve. Consider first the initial condition given by a $C^k$-curve $\gamma \colon [0,1] \rightarrow \Lf (\mathcal{N})$. The Lie algebra $\Lf (\mathcal{N})$ has been identified as $C^\infty_c (\R \times \SSS^2)$, so we can think of elements in the Lie algebra as parameter-dependent vector fields on $\R$ (with smooth dependence on a parameter $\z \in \SSS^2$). To understand the derivative of the right translation $\rho_G(F) = F\cdot G$ where the product is the group product of $\mathcal{N}$, we apply \cite[Corollary 11.6]{Michor80} twice to \eqref{eq:comp:NU}. This shows that if we identify $X \in \Lf (\mathcal{N}) = C^\infty_c (\R \times \SSS^2)$ then 
 $$T\rho_F (X) = X \circ (F,\text{pr}_2) \text{, i.e. } T\rho_F (X)(t,\z) = X (F(t,\z),\z).$$
 In other words a $C^{k+1}$-curve $\eta \colon [0,1] \rightarrow \mathcal{N}$ solves the differential equation \eqref{eq:regularity} if the associated map $\eta^\wedge \colon [0,1] \times \R \times \SSS^2 \rightarrow \SSS^2$ solves the time and parameter dependent flow equation
 \begin{equation}\label{parameter:CRSODE}
  \frac{\partial}{\partial s} \eta^\wedge (s,t,\z) = \gamma^\wedge (s,\eta^\wedge (s,t,\z),\z).
 \end{equation}
 We recall now from \cite[Proposition 3.20]{AS15} that the evaluation map $$\ev_k \colon C^k([0,1],C^\infty_c (\R \times \SSS^2)) \times [0,1] \rightarrow C^\infty_c (\R \times \SSS^2)$$ is a $C^{\infty,k}$-map. As the evaluation $\ev$ of smooth functions is a smooth map, the chain rule \cite[Lemma 3.18]{AS15} for $C^{r,s}$-maps implies that  
 \begin{align*}
H \colon [0,1] \times \left(C^k([0,1] ,C^\infty_c (\R \times \SSS^2) \times (\R \times \SSS^2))\right) &\rightarrow \R,\\
(s,(\gamma,(u,\z))) &\mapsto \ev(\ev_k (\gamma,s),(u,\z))=\gamma(s)(u,\z)
   \end{align*} is a mapping of class $C^{k,\infty}$. Hence the right hand side of \eqref{parameter:CRSODE} is a $C^{k,\infty}$-mapping in its entries and we can apply the solution theory for differential equation whose right-hand side is of class $C^{r,s}$.
 
 \begin{lem}\label{solution:ODE}
  Let $k \in \N_0$. For every $\gamma \in C^k([0,1],C^\infty (\R \times \SSS^2))$ the differential equation \eqref{parameter:CRSODE} admits a unique solution $\eta^\wedge_\gamma \colon [0,1] \times (\R \times \SSS^2) \rightarrow \R$ which is of class $C^{k+1,\infty}$. Moreover, we obtain a map of class $C^{r+1,\infty}$ via 
  $$[0,1]\times \left((\R \times \SSS^2) \times C^k([0,1],C^\infty (\R \times \SSS^2))\right) \rightarrow \R , (s,(u,\z),\gamma) \mapsto \eta^\wedge_\gamma (s,u,\z).$$
 \end{lem}

 \begin{proof}
  Instead of the differential equation \eqref{parameter:CRSODE} we consider directly the differential equation
  \begin{equation}\label{more:general}
   \frac{\partial}{\partial s} \eta (s,u,\z) = H(s,\gamma,\eta(s,u,\z),\z)=\gamma(s)(\eta(s,u,\z),\z)
  \end{equation}
 where we regard $\gamma$ and $\z$ as parameters on which the right hand side $H$ depends smoothly. We have already seen that $H$ is a $C^{k,\infty}$-map (where the $k$ is with respect to the time variable $s$).
 We can now apply \cite[Theorem 5.6]{AS15} which ensures that for every parameter $\gamma,\z$ there exists a unique solution $\eta_\gamma^\wedge (s,u,\z)$ on some time interval (a priori depending on the parameters) around $0$. Note that this is exactly the differentiable dependence on parameters and time we claimed for the solution in the statement of the lemma. Since $\gamma(s)(\cdot,\z)$ is a compactly supported time-dependent vector field on $\R$ for every choice of the parameters $\gamma, \z$, the usual argument (see e.g.~\cite[Theorem 9.16]{Lee13}) shows that the solution exists on all of $[0,1]$.  
 \end{proof}
 
 We have now constructed candidates for the solution of the regularity problem for $\mathcal{N}$. Now these candidates need to be identified with smooth mappings taking values in the manifold of mappings.
 
 \begin{prop}\label{Nisregular}
  The Lie group $\mathcal{N}$ is $C^r$-regular for all $r \in \N_0$.
 \end{prop}

 \begin{proof}
  Fix $r \in \N_0$ and $\gamma \in C^r([0,1],C^\infty_c (\R \times \SSS^2))$. From \Cref{solution:ODE} we obtain a solution $\eta_\gamma^\wedge$ of \eqref{parameter:CRSODE}, 
  By construction we have for fixed $s \in [0,1]$ that $\eta_\gamma^\wedge (s,\cdot) \in C^\infty_c(\R \times \SSS^2)$, whence we can define $\eta_\gamma \colon [0,1] \rightarrow C^\infty (\R \times \SSS^2), \eta_\gamma (s) \coloneq \eta_\gamma^\wedge (s,\cdot)$. 
  Note that $C^\infty_c (\R \times \SSS^2) = \lim_{\rightarrow } C^\infty_L (\R \times \SSS^2)$, where $L$ runs through the compact subsets of $\R \times \SSS^2$. The inductive limit is compactly regular by \cite[4.7.8]{Michor80}. This means that since $\gamma \colon [0,1] \rightarrow C^\infty_c (\R \times \SSS^2)$ is continuous with compact image, the image of $\gamma$ is already contained in a step of the directed system. Thus for $\gamma$ there exists a compact set $L_1 (\gamma) \subseteq \R$ such that $\gamma (s)|_{\R \times \SSS^2 \setminus L_1(\gamma) \times \SSS^2} \equiv 0$ for all $s \in [0,1]$. Note that since the initial condition for the differential equation \eqref{more:general} is $\gamma (0)(u,\z) = u = p(u)$, we see that
  $$|\eta_\gamma(s)(u,\z) - u| = \left|\int_0^s \frac{d}{ds} \eta_\gamma (s)(u,\z)\mathrm{d} s\right| \leq \int_0^s |\gamma(s)(\eta_\gamma (s)(u,\z))|\mathrm{d} s$$
  is bounded by $\sup_{s,u,\z}|\gamma(s)(u,\z)| <\infty$. So if $R \coloneq \sup_{s,u,\z}|\gamma(s)(u,\z)|$ and $L_1 = [a,b]$, then outside of the compact set $L_R \coloneq [a-R,b+R]$ the map $\eta_\gamma (s)|_{\R \times \SSS^2 \setminus (L_R\times \SSS^2)}$ coincides with $p$ for all $s\in [0,1]$.\smallskip
  
  \textbf{Step 1:} \emph{$\eta_\gamma$ is a $C^{r+1}$-curve to $C^\infty (\R \times \SSS^2)$.}
  From the preliminary considerations we see that $\eta_\gamma-p$ takes its image in $C^\infty_{L_R \times \SSS^2} (\R \times \SSS^2)$. As the mapping 
  $$C^\infty_{fS} (\R \times\SSS^2) \rightarrow C^\infty_{fS} (\R \times\SSS^2),F \mapsto  F+p$$ is smooth (see \cite[Remark 4.11]{Michor80} and note that this even restricts on the component of $p$ to a manifold chart), it suffices to prove the $C^{k+1}$-property of $\eta_\gamma-p$ as a mapping to $C^\infty_{L_R \times \SSS^2} (\R \times \SSS^2)$. However, $\eta_\gamma^\wedge \colon [0,1] \times (\R \times \SSS^2) \rightarrow \R$ is a $C^{k+1,\infty}$-mapping due to \Cref{solution:ODE}. Subtracting $p$, we can exploit the exponential law \cite[Theorem B]{AS15} to see that the mapping $\eta_\gamma -p \colon [0,1] \rightarrow C^\infty_{L_R \times \SSS^2} (\R \times \SSS^2)$ is a $C^{k+1}$-map. The crucial point here is that the space on the right hand side is endowed with the compact open $C^\infty$-topology which coincides on the subspace with the fine very strong topology (see e.g.~\cite[Remark 4.5]{HaS17}). Thus $\eta_\gamma$ is a $C^{k+1}$-map. \smallskip
  
  \textbf{Step 2:} \emph{$\eta_\gamma$ is a $C^{k+1}$-map with image in $\mathcal{N}$}. By construction we have $\eta_\gamma (0)=p$ and for every fixed $\z \in \SSS^2$ the map $\eta_\gamma^\wedge (s,\cdot,\z)$ is in $\Diff(\R)$ by the flow property. Now $p(\cdot,\z) \in \Diff^+(\R)$ and $s \mapsto \gamma_\eta (s)(\cdot,\z) \in C^\infty_c (\R)$ is continuous for every $\z$, whence also as a curve to $\Diff (\R)$. We deduce that $\gamma_\eta (s)(\cdot,\z) \in \Diff^+(\R)$ for every $s,\z$ as it is a continuous curve starting in $\Diff^+ (\R)$. \smallskip
  
  \textbf{Step 3:} \emph{$C^0$-semiregularity and the map $\evol$.}
  From our discussion of the differential equation governing regularity of $\mathcal{N}$, we see that $\eta_\gamma$ is a solution for \eqref{eq:regularity} for the initial value $\gamma$. Thus $\mathcal{N}$ is $C^{k}$-semiregular.
  To establish $C^0$-regularity of the group $\mathcal{N}$, we consider the map 
  $$\evol \colon C([0,1],C^\infty_c (\R \times \SSS^2)) \rightarrow \mathcal{N} \subseteq C^\infty (\R\times \SSS^2),\quad \gamma \mapsto ((u,\z)\mapsto \eta_\gamma^\wedge(1,u,\z)),$$
  where again $\eta_\gamma^\wedge$ solves \eqref{parameter:CRSODE}. In view of \cite[Lemma 3.1]{hg2015}, the group $\mathcal{N}$ will be $C^0$-regular if $\evol$ is smooth.\smallskip
  
  \textbf{Step 4:} \emph{$e_K \colon C([0,1],C_K^\infty (\R \times\SSS^2))\rightarrow C^\infty_c(\R \times\SSS^2),\gamma\mapsto \evol(\gamma)-p$ is smooth for every $K \subseteq \R \times \SSS^2$ compact.}
  Consider for $R >0$ the open set 
  $$O_R \coloneq \{\gamma \in C([0,1],C_L^\infty (\R \times\SSS^2)) \mid \sup_{(s,u,\z) \in [0,1] \times \R \times \SSS^2} |\gamma(s)(u,\z)| < R\}.$$
  As the $O_R$ exhaust $C([0,1],C_L^\infty (\R \times \SSS^2))$, it suffices to prove smoothness of $e_K$ on every $O_R$. To this end, we recall from \Cref{solution:ODE} that the associated map 
  $$e_K^\wedge \colon O_R \cap C([0,1],C_K^\infty (\R \times \SSS^2)) \times \R \times \SSS^2 \rightarrow \R , (\gamma,u,\z) \mapsto \eta_\gamma(1)(u,\z)$$
  is smooth. We now need to create a situation where the exponential law can be applied. As the compact set $K$ is contained in some compact set $L \times \SSS^2$ we proceed as in Step 1: For every $\gamma \in O_R$, the solution $\eta_\gamma^\wedge$ to \eqref{parameter:CRSODE} takes its image in $C^\infty_{L_R \times \SSS^2}(\R \times\SSS^2)$ for a compact set $L_R$ only depending on $R$. Hence we deduce that there is a compact subset $K_R \subseteq \R \times \SSS^2$ such that $\evol(\gamma)-p$ takes its image in $C_{K_R}^\infty (\R \times \SSS^2)$. Applying now the exponential law \cite[Theorem B]{AS15}, the smoothness of $e_K$ follows. \smallskip
  
  \textbf{Step 5:} \emph{$\mathcal{N}$ is $C^{0}$-regular.} To prove that $\evol$ is smooth, we exploit that $C^\infty_c (\R \times \SSS^2) = \lim_{\rightarrow} C^\infty_L (\R \times \SSS^2)$ is a compactly regular inductive limit. Thus Mujica's theorem \cite{Muj83} yields an isomorphism 
  $$C([0,1],C^\infty_c (\R \times \SSS^2) ) = C([0,1],\lim_{\rightarrow} C^\infty_L (\R \times \SSS^2)) \cong \lim_{\rightarrow} C([0,1],C^\infty_L (\R \times \SSS^2).$$
  On each step $C^\infty_L (\R \times \SSS^2)$ the topology coincides with the compact open $C^\infty$-topology and we can thus apply the exponential law \cite[Theorem A]{AS15} for $C^{0,\infty}$-mappings.
  Thus $C([0,1],C^\infty_L(\R \times \SSS^2)) \cong C^\infty_L (\R \times \SSS^2,C([0,1],\R))$ for every $L$. Passing to the limit, we deduce that the mapping 
  $$\Theta \colon C([0,1],C^\infty_c (\R \times \SSS^2)) \rightarrow C^\infty_c (\R \times \SSS^2,C([0,1],\R)),\quad \gamma \mapsto \left((u,\z ) \mapsto \gamma(\cdot )(u,\z )\right)$$
  is an isomorphism of locally convex spaces. We combine this with the fact that the image of $\evol$ is contained in the component of the unit $p \in \mathcal{N}$. Thus it suffices to establish smoothness of the map
  $$E \colon C^\infty_c (\R \times\SSS^2,C([0,1],\R)) \rightarrow C^\infty_c (\R \times\SSS^2),\quad h \mapsto \evol(\Theta^{-1}(h)) -p.$$
  Pick now $K \subseteq \R \times \SSS^2$ compact and consider the restriction of $E$ to $C^\infty_K (\R \times \SSS^2 , C([0,1],\R))$. As $\Theta^{-1}(C^\infty_{L}(\R\times\SSS^2,C([0,1],\R)))=C([0,1],C_L^\infty(\R\times \SSS^2))$, we have $E|_{C^\infty_K (\R \times \SSS^2 , C([0,1],\R))}= e_K \circ \Theta^{-1}|_{C^\infty_K (\R \times \SSS^2 , C([0,1],\R))}$, whence the restriction is smooth by Step 4.
  
  Let us now show that $E$ is a locally almost local map. To this end, we work locally on the open sets $\Theta (O_R)$ which by construction exhaust $C_c^\infty (\R \times \SSS^2,C([0,1],\R))$. Hence we fix $R>1$ and may assume that $\sup_s \sup_{(u,\z)}|F(u,\z)(s)|< R$ for every $F$ we consider. From the definition of $E$ we see that $E(F)|_K= E(G)|_K$ on some compact subset $K$ if and only if $\eta_{\Theta^{-1}(F)}$ and $\eta_{\Theta^{-1}(G)}$ coincide on $K$. To obtain such sets, we define for $n \in \Z$ the open relatively compact sets $U_n^R \coloneq ]n-R,n+2+R[ \times \SSS^2$ and $V_n \coloneq ]n,n+2[\times \SSS^2$. Clearly the resulting families are locally finite and cover $\R \times \SSS^2$. Assume now that $F$ and $G$ coincide on the open set $U_n^R$. By definition this implies that $\gamma \coloneq \Theta(F)$ and $\gamma' \coloneq \Theta(G)$ satisfy $\gamma(s)(u,\z) = \gamma'(s)(u,\Z)$ for all $s \in [0,1], (u,\z) \in U_n^R$. Moreover, $\eta_\gamma$ is the unique solution of the initial value problem
  \begin{align}\label{diffeq:loc}
   \begin{cases}
    \frac{d}{d s} \eta_\gamma (s)(u,\z) = \gamma (s)(\eta_\gamma(s)(u,\z),\z)& \forall s \in [0,1], (u,\z) \in \R \times \SSS^2,\\
    \eta_\gamma (0)(u,\z) =u .&
   \end{cases}
  \end{align}
  Now if $(u,\z) \in V_n$, we deduce from $\sup_{(s,u,\z)}|\gamma^\wedge (s,u,\z)|<R$ that the flow of \eqref{diffeq:loc} starting at $(u,\z)$ stays inside of $U_n^R$. The same observations hold for $\gamma'$. Hence $\eta_\gamma (s)(u,\z) = \eta_{\gamma'} (s)(u,\z)$ for all $(u,\z) \in V_n$ by uniqueness of solutions to the differential equations \eqref{diffeq:loc}. We conclude that $E$ restricts on every $\Theta(O_R)$ to an almost local map, whence is locally almost local.  
  
  In conclusion $E$ satisfies the the prerequisites of \Cref{Glo:paramsmooth} and is thus smooth. This concludes the proof.
 \end{proof}

\subsection{Lie group structure of the Newman--Unti group}
We have seen in the last section that the component $\mathcal{N}$ of the semidirect product $\mathcal{N} \rtimes_\tau \Lor$ comprising the Newman--Unti group is a Lie group with respect to the fine very strong topology.
Thus we can ask whether the Newman--Unti group can be turned into a Lie group, as it is the semidirect product of two Lie groups. The key is of course the action 
$$\tau \colon \mathcal{N} \times \Lor \rightarrow \mathcal{N}, (F,f) \mapsto (t,\z) \mapsto K_f^{-1}(\z) \cdot F (K_f(\z) , f(\z))$$ which needs to be smooth (with respect to the to the manifold structure we just constructed on $\mathcal{N}$). 
However, we obtain first the following negative result,

\begin{prop}\label{prop:discontNU}
 The group action $\tau \colon \mathcal{N} \times\Lor \rightarrow \mathcal{N}$ is \textbf{not} smooth in the sense of convenient analysis. So in particular it is not smooth in the Bastiani sense and thus the Newman--Unti group does not become a Lie group (neither in the Bastiani nor in the convenient setting) if we endow $\mathcal{N}$ with the Lie group structure from \Cref{prop:NisLie}.
\end{prop}

\begin{proof}
 We shall construct a smooth curve with values in $\mathcal{N} \times\Lor$ which is mapped by the action map $\tau$ to a non-smooth curve. This not only shows that the action can not be smooth in the Bastiani setting but also in the setting of convenient calculus.
 
 Identify elements in $\Lor$ with complex $2\times 2$ matrices and consider the smooth curve: 
 $$\Lambda \colon \R \rightarrow\Lor,\quad s \mapsto \Lambda_s \coloneq \begin{bmatrix} 1 + s & 1 \\ s & 1 \end{bmatrix}.$$
 Identifying $\SSS^2$ with the extended complex plane, via the stereographical projection map $\kappa \colon \SSS^2 \rightarrow \eC$, we define $\phi_s \coloneq \kappa^{-1}\circ \Lambda (s) \circ \kappa$. This yields for $\zeta = \kappa (\z)$ the relation $K_{\phi_s}(\z) = \frac{1+\lVert \zeta\rVert^2}{\lVert (1+s) \zeta + 1\rVert^2 + \lVert s \zeta + 1\rVert^2}$. To construct the desired curve with values in the product manifold $\mathcal{N} \times \Lor$ we consider the element $n \colon \R \times \SSS^2, (t,\z) \mapsto t+1$ of $\mathcal{N}$ and define the smooth curve $c(s) \coloneq (n , \phi_s)$. We will now show that the curve
 $\tau \circ c$ with values in $\mathcal{N}$ is not smooth. A trivial computation yields $\tau \circ c (s)(t,\z) = t + K_{\phi_s}( \z)^{-1}$ and for $\z =\kappa^{-1}(1)$ we obtain 
 $$\tau \circ c (s) (t,\kappa^{-1}(1)) = t + \frac{(2+s)^2+(1+s)^2}{2}.$$
  Observe now that for every compact interval $[a,b] \subseteq \R$ containing more than one point it is impossible to find $K \subseteq \R \times \SSS^2$ compact such that $\tau \circ c (s)|_{(\R \times \SSS^2) \setminus K}$ is constant in $s$ as no compact subset of $\R \times \SSS^2$ contains the set $\R \times \{\kappa^{-1}(1)\}$. However, a curve with values in $C^\infty_{\text{fS}} (\R \times \SSS^2)$ which violates this condition can not be smooth by \cite[Lemma 42.5]{KaM97}. Since $\mathcal{N}$ is an open submanifold of $C^\infty_{\text{fS}} (\R \times \SSS^2)$ the curve $\tau \circ c$ cannot be smooth. 
\end{proof}

The problem identified in the proof of \Cref{prop:discontNU} is that the action does not respect the convergence in $\mathcal{N}$. Namely, we picked an element $n$ which is not contained in the same connected component as the identity element $p \in\mathcal{N}$ and thus did not coincide with $p$ outside of any compact subset of $\R \times \SSS^2$. We will now show that this is the only defect of this action. In other words, if we restrict from $\mathcal{N}$ to the connected component 
$$\Ncvs \coloneq \{F \in \mathcal{N} \mid \exists K \subseteq \R \times \SSS^2 \text{ compact, such that } (F-p)(u,\z)=0 , \forall (u,\z) \not \in K\},$$
then we observe the following.

\begin{lem}
 The restriction of $\tau$ to the connected component $\Ncvs \subseteq \mathcal{N}$ yields a group action 
 $$\tau_0 \colon \Ncvs \times \Lor \rightarrow \Ncvs.$$ 
\end{lem}

\begin{proof}
 We have to show that $\tau_0 (F,f) \in \Ncvs$ if $F \in \Ncvs$. To this end, pick $T>0$ such that outside of the compact set $K \coloneq [-T,T] \times \SSS^2$ we have $F (t,\z) = t = p(t,\z)$ if $|t|>T$. Since $\SSS^2$ is compact, there is $M \coloneq \inf_{\z} K_f (\z) > 0$. Now if $|t|> T/M$ we have $F(K_f(\z)t,\z) = K_f(\z)t$, whence $\tau_0(F,f) (t,\z) = K_f^{-1}(\z)F(K_f(\z)t,f(\z))= t = p(t,\z)$ for all such $t$ and $\tau_0 (F,f) \in \Ncvs$.  
\end{proof}

For the restricted action we have ruled out the pathology exploited in \Cref{prop:discontNU}. We will show now that the restricted action is smooth, whence it yields a Lie group structure on the restricted semidirect product.

\begin{thm}\label{thm:NUunit_is_Lie}
 The connected component of the unit $\NUovs \coloneq \Ncvs \rtimes_{\tau_0} \Lor$ of the Newman--Unti group becomes a Lie group with respect to the submanifold structure $\Ncvs \subseteq C^\infty_{\text{fS}} (\R \times \SSS^2)$. 
\end{thm}

\begin{proof}
 It suffices to prove that the action $\tau_0$ is smooth. Note that the map 
 $$\Phi \colon C^\infty (\R\times \SSS^2) \supseteq \Ncvs \rightarrow C^\infty_c (\R \times \SSS^2) ,\quad F \mapsto F-p$$
 is a chart when restricted to the open set $\Ncvs$. We write $\Omega \coloneq \Phi(\Ncvs)$ for the open image of $\Phi$ in $C_c^\infty (\R \times \SSS^2)$ and see that it suffices to show that  
 \begin{align*}
  \tilde{\tau}\colon \Omega \times \Lor \rightarrow C^\infty_c(\R \times \SSS^2), \quad  (\alpha,f) \mapsto \left((t,\z) \mapsto K_f^{-1}(\z) \alpha (K_f(\z)t,f(\z))\right)
 \end{align*}
is smooth. We will proceed in two steps and verify the prerequisites of \Cref{Glo:paramsmooth}. Before we proceed, it is useful to construct neighborhoods in $\Lor$ bounding the maximal conformal factor: Since $\SSS^2$ is compact, for every $f \in \Lor$ the constant $\sup_{\z \in \SSS^2} K_f (\z)$ is finite. As the conformal factor $K_f (\z) $ is continuous (even smooth) in $(f,\z) \in \Lor \times \SSS^2$, the set $O_R = \{g \in \Lor \mid \sup_{\z \in \SSS^2} K_g (\z) < R\} \subseteq \Lor$ is open. Moreover $\Lor = \bigcup_{R\geq 1} O_R$. If $L \subseteq \R \times \SSS^2$ is a compact set, we can find $T_L > 0$ such that the compact set $[-T_L,T_L] \times \SSS^2$ contains $L$. Enlarging the set even further, we define the compact set $L_R \coloneq [-RT_L,RT_L] \times \SSS^2$ which again contains $L$ for all $R\geq1$. .

 \textbf{Step 1:} \emph{$\tilde{\tau}$ restricts to a smooth map on $\Phi(\Ncvs) \cap C^\infty_L (\R \times \SSS^2) \times \Lor$.}
  Clearly it suffices to prove the claim for every open set $\Omega \cap C^\infty_L (\R \times \SSS^2) \times O_R$ where $R\geq 1$. Consider now the map  associated to the restriction $\tilde{\tau}_R$ of $\tilde{\tau}$ given by 
  $$\tilde{\tau}_R^\vee \colon (\Omega \cap C^\infty_L (\R \times \SSS^2)) \times O_R \times \R \times \SSS^2 \rightarrow \R,\quad (\alpha,f,u,\z)\mapsto K_f^{-1}(\z)\cdot \alpha (K_f (\z)u,f(\z)).$$
  By construction, $\tilde{\tau}_R^\vee$ vanishes outside of $(\Omega \cap C^\infty_L (\R \times \SSS^2)) \times O_R \times L_R$. Since $L_R$ is compact, we can apply \cite[Lemma C.3]{AaS19} to see that $\tilde{\tau}_R$ will be smooth if $\tilde{\tau}_R^\vee$ is smooth. However, $K_f (\z), K_f^{-1}(\z)$ and the canonical action $\Lor \times \SSS^2 \rightarrow \SSS^2$ are all smooth. Thus the smoothness of $\tilde{\tau}_R^\vee$ follows directly from the smoothness of the evaluation $\ev \colon C^\infty_L (\R \times \SSS^2) \times \R \times \SSS^2 \rightarrow \R$ (cf.\ \cite[Proposition 3.20]{AS15}). 
  
 \textbf{Step 2:} \emph{The mapping $\tilde{\tau}$ is almost local.}
 Let $R \in \N$ and define the open set $Q_R \coloneq \{f \in O_R \mid \forall \z \in \SSS^2, 1/R < K_f (\z) < R\}$. Then $\Lor = \bigcup_{R \in \N} Q_R$. We fix $R \in \N$ and define families of relatively compact open sets as follows: Let $A$ be the set of all integers and define $V_a \coloneq ]a-2,a+2[ \times \SSS^2$ for $a \in A$. Then the $V_a$ form a locally finite family of relatively compact sets covering $\R \times \SSS^2$. To define the sets $U_a$ we need to distinguish several cases:
 For the first case, assume that $(0,\z) \in V_a$. We define $U_s \coloneq ]\min (R(a-2),-2R),\max(R(a+2),2R)[\times \SSS^2$. By construction, we have that if $(u,\z) \in V_a$, for every $f\in Q_R$ we have $|K_f (\z) u| \leq R|u|$, whence $(K_f(\z)u,f(\z)) \in U_a$ for all $f \in Q_R$. 
 For the other cases, assume that $(0,\z) \not \in V_a$. Define 
 $$U_a \coloneq \begin{cases}
                 ](a-2)/R,R(a+2)[\times \SSS^2 & a-2\geq 0\\
                 ]R(a-2), (a+2)/R[ \times \SSS^2 & a+2 \leq0
                \end{cases}
$$
Thus for $a-2\geq 0$, $(u,\z) \in V_a$ and $f \in Q_R$ we obtain the inequalities 
$$\frac{a-2}{R} <\frac{u}{R} \leq K_f (\z)u \leq R u < R(a+2)$$
and thus $(K_f (\z)u,f(\z)) \in U_a$ by construction. Similarly, we obtain the same in the case $a+2\leq 0$. In conclusion, we have constructed a locally finite family $(U_a)_a$ of relatively compact sets which covers $\R \times \SSS^2$. 
If now $f \in Q_R$ and $\alpha|_{U_a} \equiv \beta|_{U_a}$ then we deduce from the construction of the locally finite families that for $(u,\z) \in V_a$ we have 
\begin{align*}
 \tilde{\tau}(\alpha,f)(u,\z) = K_f^{-1} (\z)\alpha(K_f(\z)u,\z)=K_f^{-1} (\z)\beta(K_f(\z)u,\z)=\tilde{\tau}(\beta,f)(u,\z).a
\end{align*}
This shows that $\tilde{\tau}$ is almost local.
 
 Combining Steps 1 and 2, the smoothness of $\tilde{\tau}$ follows from \Cref{Glo:paramsmooth}.
\end{proof} 

While this indeed constitutes a Lie group structure on the unit component of the NU group (seen as a submanifold of $C^\infty_{fS}(\R\times \SSS^2)\times \Lor$), we note that this subgroup does not accommodate the image of the canonical inclusion of the BMS group. Recall that the inclusion was given by
$$I \colon \BMS \rightarrow \NU,\quad (F,\phi) \mapsto (p+F\circ \text{pr}_2,\phi).$$
Then the formula for the first component shows that the only element which gets mapped by $I$ into the subgroup $\Ncvs \rtimes \Lor$ is the identity supertranslation $F \equiv 0$.

Note that since $\Ncvs$ is the unit component of the Lie group $\mathcal{N}$ we have $\Lf (\Ncvs) = \Lf (\mathcal{N}) = C_c^\infty (\R \times \SSS^2)$ with the bracket computed in \Cref{prop:LieAlg-N}. 
Now $\NUovs$ is a semidirect product of $\Ncvs$ and $\Lor$, whence its Lie algebra is given as the semidirect product
$$\Lf (\NUovs) = C_c^\infty (\R \times \SSS^2) \rtimes_{d\tau} \Lf(\Lor),$$
where $d\tau$ is the derived action of $\tau$. Moreover, as $\mathcal{N}$ is $C^0$-regular by \Cref{Nisregular}, $\Lor$ is $C^0$-regular as a finite-dimensional Lie group and $C^0$-regularity is an extension property we obtain:

\begin{cor}\label{cor:NUunit_is_regular}
 The Lie group $\NUovs$ is $C^0$-regular. 
\end{cor}

As a consequence of $C^0$-regularity we obtain the following properties of $\NUovs$:

\begin{lem}
 The strong Trotter and the strong commutator property hold for the Lie groups $\mathcal{N}$ and $\NUovs$.
\end{lem}

\begin{proof}
 For any $C^0$-regular Lie group the strong Trotter property holds due to \cite[Theorem 1]{Han20}. Further, it is known that the strong Trotter property implies the (strong) commutator property \cite[Theorem H]{glo15}.
\end{proof}

On the other hand we have the following negative results:

\begin{prop}\label{prop:negativeLie}
 The Lie groups $\mathcal{N}$ and $\NUovs$ are not analytic Lie groups. 
\end{prop}

\begin{proof}
 It is well known that the group $\Diff^+(\R)$ is not convenient real analytic by \cite[43.3 Remarks]{KaM97}. As $\mathcal{N}$ is essentially a parametrized version of $\Diff^+(\R)$ these results carry over to $\mathcal{N}$. So we only need to notice that if a Lie group is not convenient real analytic, it can not be real analytic in the Bastiani sense.  
 
 For the group $\NUovs$ the corresponding results follow now at once from the semidirect product structure. 
 Note that as an alternative, we could also have deduced the lack of an analytic structure from the failure of $\tau$ to be analytic (this works exactly as the corresponding proof for the BMS group, see \cite[Proposition 3.9]{Prinz_Schmeding_1}). 
\end{proof}
Note that as a consequence of \Cref{prop:negativeLie} for both Lie groups the Baker--Campbell--Hausdorff series does not provide a local model for the multiplication on the Lie algebra. 
It is yet unclear as to whether the Lie group $\mathcal{N}$ and $\NUovs$ are locally exponential, i.e.~that their Lie group exponential induces a local diffeomorphism between a $0$-neighborhood in the Lie algebra and a unit neighborhood in the Lie group. Let us note that if $\mathcal{N}$ is not locally exponential this property carries over to $\NUovs$ due to the semidirect product structure (however, if $\mathcal{N}$ was locally exponential it is not easy to deduce local exponentiality of $\NUovs$). Since $\Diff^+(\R)$ is not locally exponential (see e.g.~\cite{KaM97}) we strongly suspect that also $\mathcal{N}$ is not locally exponential. A proof of this statement would require a detailed analysis and adaption of the arguments for the non local exponentiality of $\Diff (\R)$ in \cite{Gra88}. This is beyond the scope of the current paper, but we pose the following

\begin{conj} \label{conj:NU_not_locally_exponential}
The Lie groups $\mathcal{N}$ and $\NUovs$ are not locally exponential.
\end{conj}

\section{Conclusion}

We have studied the NU group from the viewpoint of infinite-dimensional Lie group theory. In particular, we have discussed several possible topologies, which turn the supertranslation part \(\mathcal{N}\) of the NU group either into a topological group (\propref{prop:NUtopgp}) or into a Lie group (\propref{prop:NisLie}). However, we have also shown that only the connected component \(\NUovs\) of the identity of the complete NU group becomes a Lie group, as the group action is not smooth on the full NU group (\propref{prop:discontNU} and \thmref{thm:NUunit_is_Lie}). Furthermore, we have shown that both the Newman--Unti supertranslation group \(\mathcal{N}\) as well as the connected component of the identity \(\NUovs\) in the manifold topology of the Newman--Unti group, are regular in the sense of Milnor (\propref{Nisregular} and \colref{cor:NUunit_is_regular}). Moreover, we have shown that both of these Lie groups are not analytic (\propref{prop:negativeLie}). Moreover, we have shown that while the BMS group can be embedded into the NU group as a topological group (\lemref{lem:BMS_into_NU_topological}), this inclusion is not possible on the level of Lie groups, contrary to their Lie algebras, which split as a direct sum (cf.\ Equation~\eqref{eqn:decomposition_nu-bms-conf}). Finally, we remark our previous article \cite{Prinz_Schmeding_1} in which we have studied the BMS group from a Lie theoretic perspective.

\begin{appendix}
\section{Auxiliary results for \Cref{sect:normalsubgrp}}\label{app:aux}
In this appendix we compile some auxiliary results used in the construction of  the Lie group structure for the Newman--Unti group. These results are easy extensions of well known results. However, we were not able to find them in citable form in the literature, whence they are compiled here for the readers convenience. We start with a technical lemma:

\begin{lem}\label{lem:step2}
 If $F \in C^\infty_{fS} (\R \times K)$ for $K$ a compact manifold such that $\frac{\partial}{\partial u} F(u,\z) >0,\ \forall (u,\z) \in \R \times \SSS^2$, then there is an open $F$-neighborhood $O \subseteq C^\infty_{fS}(\R \times K)$ of mappings with this property.
\end{lem}

\begin{proof}
  Pick a finite atlas $(\psi_j,U_j)_{j=1,\ldots k}$ of $K$ together with compact sets $A_j \subseteq U_j$ such that the $A_j$ cover $K$. Choose a locally finite family $(K_i)_{i \in \N}$ of compact subsets which cover $\R$. 
  For every $i \in \N,j=1,\ldots ,k$ there is $\varepsilon_{i,j} >0$ such that for $G \in C^{\infty} (\R \times K)$ the condition on the directional derivative $d(F-G)$
\begin{equation}\label{eq:loc:estimate}
 \sup_{(u,\z) \in K_i \times A_j}\left|d\left((F-G)\circ (\id_R \times \psi_j^{-1})\right) (u,\z;1,0)\right| < \varepsilon_{i,j} 
\end{equation}
implies that $G(\cdot,\z)$ satisfies $\frac{\partial}{\partial u} G(u,\z) >0$ for all $(u,\z) \in K_i \times A_j$. By construction, the family of compact sets $(K_i \times A_j)_{i,j}$ is locally finite and we see that there is a basic open neighborhood $O_F$ of $F$ in $C^\infty_{\text{fS}} (\R \times K)$ consisting only of mappings which satisfy \eqref{eq:loc:estimate} (compare \cite[Definition 1.6]{HaS17}).
\end{proof}

The next result is a parametrized version of \cite[Lemma 2.1.3]{Hir76} (the proof is completely analogous apart from the presence of another parameter).

\begin{lem}\label{para:emb}
 Let $K$ be a compact manifold and $F \in C^\infty (\R \times K)$ be a mapping such that for every $k \in K$ the partial map $F(\cdot,k)$ is an embedding. Furthermore, we fix $W \subseteq \R$ open and relatively compact. Then there exists an open neighborhood $N_F (W)$ of $F$ in the fine very strong topology such that for every $G \in N_F(W)$ the partial maps $G(\cdot, k) \colon \R \rightarrow \R, k\in K$ satisfy
 \begin{enumerate}
  \item $G(\cdot, k)$ is an immersion,
  \item $G(\cdot,k)|_W$ is an embedding.
 \end{enumerate}
\end{lem}

\begin{proof}
\Cref{lem:step2} provides an open neighborhood $O$ of $F$ in the fine very strong topology consisting only of mappings $G$ such that the partial maps $G(\cdot, k)$ are immersions (note that in the proof we assumed the stronger condition that the derivative of the partial maps $F(\cdot,k)$ is everywhere positive. It is easy to see that the same argument holds for mappings whose partial mappings have everywhere negative derivative). 
 We can shrink $O$ such that every $G \in O$ restricts to an embedding $G(\cdot, k)$ on $W$ for all $k \in K$. Assume to the contrary that it were not possible to construct such a neighborhood. Then there must be a sequence $G_N \rightarrow F_N$ of mappings $G_N \in C^\infty (\R \times K)$, where convergence of the sequence is in the sense of uniform convergence of the function and its first derivative on the compact set $\overline{W} \times K$ (actually convergence here is even stronger, but we only need the convergence on the compact set). In other words
\begin{align}
 \lim_{n \rightarrow \infty}\sup_{i=0,1}\sup_{(x,k) \in \overline{W} \times K}\left| \frac{\partial^i}{\partial x} G_n (x,k)- \frac{\partial^i}{\partial x^i} F (x,k)\right| = 0. \label{deriv:eq}
\end{align}
By assumption there exist distinct $a_{n,k} , b_{n,k} \in W$ such that $G_n (a_{n,k},k) = G_n (b_{n,k},k)$ for every $n \in \N, k \in K$. Now $\overline{W}$ is compact, whence $a_{n.k} \rightarrow a_k$ and $b_{n,k} \rightarrow b_k$ for some $a_k,b_k \in \overline{W}$. Evaluating \eqref{deriv:eq} for $i=0$ this entails $F(a_k,k) = F(b_k,k)$ for every $k \in K$. Hence $a_k=b_k$.
Choosing subsequences if necessary, we may assume that the sequence of unit vectors 
$$v_{n,k} \coloneq \frac{a_{n,k} - b_{n,k}}{|a_{n,k}-b_{n,k}|}$$ 
converges (either to $1$ or $-1$). Now we take a Taylor expansion in the $x$-variable while we treat the $k$-variable as a parameter. Note that the resulting function in $x$, as well as its derivative in $x$ are both uniformly continuous in the parameter $k$ (using compactness of $K$).
Now by uniformity of the Taylor expansion (with remainder in integral form) we see that 
$$\lim_{n \rightarrow \infty} \sup_{(x,k) \in \overline{W} \times K} \frac{\left| G_n(a_{n,k},k)-G_n (b_{n,k},k) - \frac{\partial}{\partial x} F(b_{n,k},k)(a_n-b_n)\right|}{|a_n-b_n|} = 0.$$
Thus $\frac{\partial}{\partial x} F(b_{n,k},k)v_{n,k} \rightarrow 0$ for every $k \in K$. Note that the absolute values of this sequence converge to $|\frac{\partial}{\partial x} F(b_k,k)|$ and this contradicts that $F(\cdot,k)$ is an immersion for every $k \in K$. 
\end{proof}
\begin{rem}\label{rem:intersect}
 From the proof of \Cref{para:emb} it is clear that the condition that $G(\cdot,k)$ restricts to an embedding on $W$ for every $k \in K$ requires us only to control the function and its derivative on the compact set $\overline{W} \times K$. Hence if we have a locally finite family $(K_i)_{i \in \N}$ of compact sets covering $\R$, we can intersect countably many of the neighborhoods constructed in \Cref{para:emb} and still retain an open set in the fine very strong topology. This is due to the fact that the fine very strong topology admits to control a function and its derivatives on any locally finite family of compacta simultaneously, cf.~\cite{HaS17} and see also \cite[Proof of Theorem 2.1.4]{Hir76}.
\end{rem}

This enables us to prove \Cref{prop:tech_NUgroup}, i.e.~$\mathcal{N}$ is an open subset of $C^\infty_{\mathrm{fS}} (\R \times \SSS^2)$. While the proof is similar to the classical proof, it carries an additional parameter.

\begin{proof}[Proof of \Cref{prop:tech_NUgroup}]
 Adapting \cite[Theorem 1.7]{Hir76} let us recall that a mapping $f \in C^\infty (\R)$ is an element of $\Diff^+ (\R)$ if and only if $f$ satisfies the following conditions
 \begin{enumerate}
  \item $\frac{\mathrm{d}}{\mathrm{d}u} f(u) >0,\ \forall u \in \R$ (i.e.~ $f$ is an orientation preserving local diffeomorphism)
  \item $f$ is proper,
  \item $f$ is injective.
 \end{enumerate}
To see that this is true, we need only notice that these conditions entail that $f$ is a diffeomorphism: Local diffeomorphisms are open maps, i.e.~the image of $f$ is open. Proper maps have closed image, whence by connectedness of $\R$, $f$ must be surjective. Hence $f$ is a bijective local diffeomorphism and thus a diffeomorphism.

Let us now show that for each $F \in \mathcal{N}$ there is an open neighborhood which contains only mappings $G \colon \R \times \SSS^2 \rightarrow \R$ whose partial mappings $G(\cdot, \z)$ satisfy 1.-3.
Indeed, we already know from \Cref{lem:step2} that $F$ admits an open neighborhood $O$ such that every partial map $G(\cdot,\z)$ for $G\in O$ satisfies 1.\\[.5em]

\textbf{Step 1:} \emph{A neighborhood of proper maps}. The set of proper maps $\mathcal{P} \subseteq C^\infty_{\mathrm{fS}} (\R \times \SSS^2)$ is open by \cite[Theorem 1.5]{Hir76}. For every $\z \in \SSS^2$ the inclusion $i_\z \colon \R \rightarrow \R \times \SSS^2, u \mapsto  (u,\z)$ is a smooth and proper map. So for every $G \in \mathcal{P}$ the partial map $G(\cdot , \z) = G \circ i_\z$ is proper. Since $\mathcal{N} \subseteq \mathcal{P}$ by \Cref{lem:properness}, we intersect $O$ and $\mathcal{P}$ to obtain an open neighborhood of mappings satisfying 1.-2.

\textbf{Step 2:} \emph{An $F$-neighborhood $O \subseteq C^\infty (\R \times \SSS^2)$ such that $G(\cdot, \z)$ is injective for $G \in O$}. 
The argument is a variant of \cite[Theorem 2.1.4]{Hir76}: Fix the locally finite family $K_i = [i-1,i+1], i\in \mathbb{Z}$ of compact sets covering $\R$. Then the open sets $U_i = ]i-2,i+2[, i \in \mathbb{Z}$ contain the compact sets $K_i \subseteq U_i$ and form a locally finite family. As $F(\cdot, \z)$ is a proper embedding (so in particular a closed injective map), we have for each $\z \in \SSS^2$ two closed disjoint sets $F(K_i, \z)$ and $F(\R \setminus U_i , \z)$. By construction, the first set is compact and connected, while the second set is closed and consists of two connected components. 
The minimal distance $3r_{i,\z}$ (for some $r_{i,\z}>0$ between the closed sets is realized by one of the pairs $(F(i-2,\z),F(i-1,\z))$ or $(F(i+1,\z),F(i+2,\z))$.

We define now two disjoint open sets 
$$A_{i,\z} \coloneq  \{x \in \R \mid \sup_{y \in F(K_i,\z)} |x-y| < r_{i,\z}\} \text{ and } B_{i,\z} \coloneq \{x \in \R \mid \sup_{y \in F(K_i,\z)}|x-y| > 2r_{i,\z}\}.$$
By construction, $F(K_i,\z) \subseteq A_{i,\z}$ and $F (\R\setminus U_i , \z) \subseteq B_{i,\z}$. Using continuity of $F$ and \cite[3.2.10 The Wallace Theorem]{Eng89}, there is a compact neighborhood $L_{i,\z}$ of $\z$ such that $F(K_i \times L_{i,\z}) \subseteq A_{i,\z}$. We will now shrink $L_{i,\z}$ such that also $F((\R \setminus U_i) \times L_{i,\z}) \subseteq B_{i,\z}$ holds. By construction, $\frac{\partial}{\partial u} F(u,s) >0$  for all $s \in \SSS^2$, whence 
\begin{align}\label{eq:boundaryreq}
 F(i-2,s), F(i+2,s) \in B_{i,\z} \text{ implies } F(\R \setminus U_i,s) \subseteq B_{i,\z}. 
\end{align}
Applying again continuity of $F$ and Wallace theorem, we can shrink $L_{i,\z}$ such that every $s \in L_{i,\z}$ satisfies \eqref{eq:boundaryreq}. Define now the open $F$-neighborhood 
$$N_{F, i,\z} \coloneq \{G \in C^\infty (\R \times \SSS^2) \mid \forall s \in L_{i,\z}, G(K_i , s) \subseteq A_{i,\z}, G(i-1,s), G(i+2,s) \in B_{i,\z}\} \cap O,$$
where $O$ is the open $F$-neighborhood from \Cref{lem:step2}. By construction, \eqref{eq:boundaryreq} implies that every $G \in N_{F,i,\z}$ satisfies $G(\R \setminus U_i,s)\subseteq B_{i,\z}$.

Now we exploit compactness of $\SSS^2$ to find finitely many $\z_{i,1}, \ldots, \z_{i,j_i}$ such that $\SSS^2 = \bigcup_{k=1}^{j_i} L_{i,\z_k}$. Intersecting the open sets $N_{F,i,\z_k}$, we find an open $F$-neighborhood of functions which satisfy the condition \eqref{eq:boundaryreq} for all $s \in \SSS^2$, but with different sets $A_{i,\z},B_{i,\z}$. Note in addition that this neighborhood only controls functions on the compact set $K_{i} \cup \{ i-2,i+2\}\times \SSS^2$.

Repeat the construction for all $i$ to obtain an open $F$-neighborhood of functions which satisfy $G(K_i,s) \subseteq A_{i,\z}$ and $G(\R\setminus U_i,s) \subseteq B_{i,\z}$ for all $i\in \mathbb{Z}$, some $\z \in \SSS^2$ and disjoint open subsets $A_{i,\z}, B_{i,\z}$. Applying \Cref{para:emb} and \Cref{rem:intersect} we can shrink the open neighborhood to obtain an open neighborhood $\Omega$ of $F$ such that every $G(\cdot,s)|_{U_i},i\in \mathbb{Z}, s\in \SSS^2$ is an embedding.
Let us show now that the partial maps for every element in $\Omega$ are injective. If $x \in K_i,y \in \R$ are two distinct points and $s \in \SSS^2$, we distinguish two cases: If $y \in U_i$, then $G(x,s)\neq G(y,s)$ as $G(\cdot,s)$ is an embedding on $U_i$. If $y \in \R \setminus U_i$, we then have $G(x,s) \in A_{i,\z}$ and $G(y,s) \in B_{i,\z}$ for some $\z$. Since these sets are disjoint, we see that also $G(x,s)\neq G(y,s)$. This concludes the proof.  
\end{proof}

It is possible to extend the proof of \Cref{prop:tech_NUgroup} to groups of smooth maps defined on $\R^n \times K$ (replacing $\R$ with a multidimensional space and $\SSS^2$ with a compact manifold). While it is apparent from the proof that we have only exploited the compactness of $\SSS^2$, the argument for $\R^n$ is slightly more involved. While groups with a compact manifold $K$ replacing $\SSS^2$ have been considered in the literature (corresponding to different spacetime dimensions), the extension to $\R^n$ does not seem to be physically relevant. We thus omit the discussion of any details to this extension.

\subsection*{Auxiliary results for the inversion in $\mathcal{N}$}
In this subsection, we provide the necessary details to establish smoothness of the inversion mapping for the $\mathcal{N}$-component of the Newman--Unti group. By construction, this is just a parameterized version of the inversion map in the Lie group $\Diff^+ (\R)$.
However, since manifolds of mappings on non-compact manifolds are somewhat delicate, there seems to be no quick way of leveraging the fact that we already know that inversion in $\Diff^+ (\R)$ is smooth. 
Instead, we have to mimic the proof for the smoothness of inversion in the diffeomorphism group as outlined in \cite{Glo05} but with an additional parameter inserted. We start with some preparation:

\begin{setup}\label{setup:subspaces}
 Let $K$ be a compact manifold and $p \colon \R \times K \rightarrow \R , (x,k) \mapsto x$. Consider the subset 
 $$G = \left\{F \in C^\infty (\R \times K) \middle| \substack{F(\cdot,k) \in \Diff^+ (\R) \forall k\in K\text{ and } \\ \exists L \subseteq \R \text{ compact with }(F-p)|_{(\R \setminus L) \times K} = 0}\right\}.$$
 Then $G$ is an open subset of $C^\infty_{fS} (\R \times K)$.
 To see this, note that \Cref{lem:step2} entails that the set $\{F \in C^{\infty}(\R\times K)\mid F(\cdot ,k) \in \Diff^+ (\R), \forall k \in K\}$ is open in the fine very strong topology. In addition, the set 
 $$\Omega_p \coloneq \{F \in C^{\infty}(\R\times K)\mid \exists L \subseteq \R \text{ compact with }(F-p)|_{(\R \setminus L) \times K} = 0\}$$ is open in this topology. Thus $G$ is open as the intersection of two open sets.
 
 Let us fix some subspaces of $C^\infty (\R \times K)$: First of all we consider the open subset
 $$C^\infty_c (\R \times K) \coloneq \{F \in C^\infty (\R\times K)\mid \exists L \subseteq \R \text{ with } F|_{(\R\setminus L) \times K} \equiv 0\}$$
 and note that the fine very strong topology turns it into a locally convex vector space.
 Furthermore, we define for every $L \subseteq \R$ compact the subset 
$$C^\infty_L (\R \times K) \coloneq \{F \in C^\infty_c (\R \times K) \mid F|_{(\R \setminus L) \times K} \equiv 0\}.$$
Then $C^\infty_L (\R \times K)$ is a closed locally convex subspace of $C^\infty_c (\R \times K)$ and one can prove that $C^\infty_c (\R \times K)$ is the locally convex inductive limit of the spaces $C^\infty_L (\R \times K)$ (partially ordered by obvious inclusion as $L$ runs through all compact subset of $\R$). Moreover, the subspace topology on $C^\infty_L (\R \times K)$ is a \Frechet\ topology, i.e.~a the space is a complete metrizable space. See \cite{Glo02} for more information.
\end{setup}

\begin{setup}\label{localsetup}
Arguing as in \Cref{lem:step2} we define another open subset of $C^\infty_{fS} (\R \times K)$: 
$$U_0 \coloneq \left\{F \in C^\infty_c (\R \times K) \middle| \sup_{(x,k) \in \R \times K}\left|\frac{\partial}{\partial x} F(x,k)\right| < 1\right\}.$$ 
Observe that the map 
 $$\Phi\colon  \Omega_p \rightarrow C^\infty_c (\R \times K) ,\quad F \mapsto F-p$$
 is a homeomorphism, mapping $G$ onto an open subset of the locally convex space. Indeed the restriction of $\Phi$ to $G$ is a chart for the manifold $G$.
 We define the open subset $G_0 = \Phi^{-1} (U_0) \cap G$ of $G$ and observe that it is a neighborhood of $p$.
 
 We now build a local model for the pointwise inversion map $F^{-1} (x,k) \coloneq (F(\cdot, k)^{-1})(x)$, by defining
 $$U_0 \cap \Phi(G) \rightarrow C_c^\infty (\R \times K), \gamma \mapsto \gamma^\ast \coloneq \Phi((\Phi^{-1}(\gamma))^{-1})$$
 Obviously inversion in the group $G$ will be smooth in a neighborhood of $p$ if and only if the local inversion map is smooth on $U_0 \cap \Phi(G) = \Phi(G_0)$.
\end{setup}

\begin{setup}
 A quick computation yields the validity of the following formulae for $\gamma, \eta \in \Phi(G_0)$:
 \begin{align}
  \Phi^{-1}(\Phi (\gamma) \circ (\Phi(\eta) \times \id_K)) = \eta + \gamma\circ ((p+\eta) \times \id_K) \label{eq:loc_mult}\\
  \gamma^\ast + \gamma\circ ((p+\gamma^\ast) \times \id_K) = 0 \text{ and } \gamma + \gamma^\ast \circ ((p+\gamma))\times \id_K) = 0, \label{eq:loc_inv}
 \end{align}
where \eqref{eq:loc_inv} is a direct consequence of \eqref{eq:loc_mult}.
\end{setup}

\begin{lem}\label{lem:res:smooth}
 For every $L \subseteq \R$ compact, the mapping 
 $$\iota_L \colon \Phi(G_0) \cap C^\infty_L (\R \times K) \rightarrow C^\infty_L (\R \times K), \quad \gamma \mapsto \gamma^\ast$$
 makes sense and is smooth.
\end{lem}

\begin{proof}
 Let us show that $\iota_L(\gamma)|_{(\R \setminus L) \times K} \equiv 0$. If we fix $k \in K$ then we need to prove that the map $\gamma^\ast (\cdot,k)$ vanishes outside of $L$ if $\gamma$ vanishes outside of $L$. Due to our definition of $U_0$ this pointwise property is a consequence of \cite[Step 3 in the proof of Lemma 5.1.]{Glo05}. Thus $\iota_L$ makes sense and takes its image in $C^\infty_L (\R \times K)$.
 To establish smoothness of $\iota_L$, we consider the associated map 
 $$\iota_L^\vee \colon \Phi(G_0) \cap C^\infty_L (\R \times K) \times (\R \times K) \rightarrow \R ,\quad (\gamma , x,k) \mapsto \gamma^{\ast} (x,k).$$
 By construction every $\gamma \in C^\infty_L (\R \times K)$ vanishes outside of the compact set $L \times K \subseteq \R \times K$. Thus \cite[Lemma C.3]{AaS19} implies that $\iota_L$ will be smooth if $\iota_L^\vee$ is smooth.
 However, for $\iota_L^\vee (\gamma) = \gamma^\ast$ we obtain from \eqref{eq:loc_inv} the implicit equation
 \begin{align}\label{eq:implicit_inverse}
  \gamma^\ast (x,k) + \gamma(x+\gamma^\ast (x,k),k) = 0,
 \end{align}
 which depends on the parameter $\gamma \in \Phi(G_0) \subseteq C^\infty_L (\R \times K)$. Now the topology of $C^\infty_L (\R \times K)$ coincides with the subspace topology induced by the compact open $C^\infty$-topology on $C^\infty (\R \times K)$ (cf.~\cite[Remark 4.5]{HaS17}). In particular, \cite{AS15} entails that the evaluation map $\ev \colon C^\infty_L (\R \times K) \times \R \times K \rightarrow \R , (\gamma,x,k) \mapsto \gamma(x,k)$ is smooth. 
 Hence the left-hand side is given by the smooth function 
 $$H \colon \Phi(G_0) \times \R \times K \times \R \rightarrow \R, \quad H(\eta,x,k,Z) = Z +\ev(\gamma, (x+Z,k)).$$
 Taking the derivative with respect to the $Z$ variable we obtain the continuous linear map
 $$d_4H (\gamma,x,k,Z;\bullet) = \bullet \cdot \left(1+ \frac{\partial}{\partial x}\gamma(x+Z,k)\right).$$
 Since $\gamma \in \Phi (G_0) \subseteq U_0$, we see that $\lVert \id_\R - d_4H (\eta,x,k,Z;\bullet)\rVert_{\text{op}} = |\frac{\partial}{\partial x}\gamma(x+Z,k)| < 1$. 
 In other words, $d_4H (\gamma,x,k,Z)$ is invertible. We can thus again apply the implicit function theorem with parameters in locally convex spaces \cite[Theorem 2.3]{hg06} to deduce that $\iota_L^\vee$ is smooth. 
\end{proof}

To extend the smoothness assertion now from the closed subspace $C^\infty_L (\R \times K)$ to all of $C^\infty_c (\R \times K)$ we employ almost local mappings, \Cref{setup:almostlocal} (here without parameter).

\begin{lem}\label{lem:almostloc}
 Let $R>0$ and $O_R$ be as in Step 2 of the proof of \Cref{prop:NisLie}. The restriction of the inversion map 
 $$\Phi (G_0) \cap O_R  \rightarrow C^\infty_c (\R \times K), \gamma \mapsto \gamma^\ast$$
 is an almost local map.
\end{lem}

\begin{proof}
Denote for $x \in \R$ by $B_r (x)$ the $r$-ball around $x$ in $\R$

 \textbf{Step 1:} \emph{For all $r>0$, $x\in \R$ and $\gamma \in O_R$ we have $B_r(x) \subseteq (p+\gamma) (B_{r+R} (x) \times K)$.}
 Let $y \in B_r (x)$ and recall that $p+\gamma = \Phi^{-1}(\gamma) \in G_0$, whence for every fixed $k \in K$ the map $\Phi^{-1} (\gamma) (\cdot,k)$ is a diffeomorphism of $\R$ there is a unique element $y_k$ such that 
 $$y = p(y_k,k)+\gamma(y_k,k) = y_k + \gamma(y_k,k), \text{ and }\quad |y_k-x| = |y-\gamma(y_k,k)-x| \leq r+R,$$  
 where we exploited that $\gamma \in O_R$. In particular, $y_k \in B_{r+R} (x)$ for all $k \in K$.
 
 \textbf{Step 2:} \emph{If $R,r >0$ and $x \in \R$, then for all $\gamma , \eta \in O_R$ such that $\gamma|_{B_{r+R}(x) \times K} = \eta|_{B_{r+R}(x) \times K}$ we have $\gamma^\ast|_{B_{r}(x) \times K} = \eta^\ast|_{B_{r+R}(x) \times K}$.}
 Let $y \in B_r(x)$. By Step 1 we find for every $k \in K$ a unique element $y_k \in B_{r+R}(x)$ such that $y_k + \gamma (y_k,k) = y$. Then from \eqref{eq:loc_inv} we deduce that
 $$\gamma^\ast (y,k) = \gamma^\ast (y_k + \gamma (y_k,k),k)= - \gamma(y_k,k) = -\eta(y_k,k)=\eta^\ast(y_k + \eta (y_k,k),k)=\eta^\ast (y,k)$$
 
 \textbf{Step 3:} \emph{Inversion is almost local}. Define for $z \in \mathbb{Z}$ the open sets $V_z \coloneq B_2 (z) \times K$, $W_z \coloneq B_3 (z) \times K$, $U_z \coloneq B_{2+R}(z) \times K$ and $X_z \coloneq B_{3+R}(z) \times K$.
 We deduce from Step 2 that the restriction of the inversion map to $O_R \cap \Phi(G_0)$ together with the families $\{U_z\}_{z\in\mathbb{Z}}$, $\{X_z\}_{z\in\mathbb{Z}}$, $\{V_z\}_{z\in\mathbb{Z}}$ and $\{W_z\}_{z \in \mathbb{Z}}$ satisfy the requirements in the definition of an almost local mapping.
\end{proof}
\end{appendix}

\addcontentsline{toc}{section}{References}
\bibliography{NU_project_v3}
\end{document}